%% file: main.tex
\documentclass[12pt]{article}
\usepackage{amsmath}
\usepackage{graphicx}
\usepackage{enumerate}
\usepackage{natbib}
\usepackage{url} 
\usepackage{amssymb}
\usepackage{amsthm}

\usepackage{wrapfig}
\usepackage{mathtools}
\usepackage{float}
\usepackage[hidelinks]{hyperref}
\usepackage{mathtools}
\usepackage{caption}
\usepackage{subcaption}
\usepackage{scalerel,stackengine}
\usepackage{relsize}
\usepackage{mathrsfs}
\usepackage{algorithm}
\usepackage{algpseudocode}
\usepackage{xcolor}
\usepackage{changebar}
\newtheorem{theorem}{Theorem}

\newtheorem{remark}{Remark}

\newtheorem{corollary}{Corollary}

\newcommand\bmat[1]{\mathbf{#1}}

\newcommand{\nacf}[0]{\operatorname{nacf}}

\newcommand{\pnacf}[0]{\operatorname{pnacf}}
\newcommand{\GNAR}[0]{\operatorname{GNAR}}

\newcommand{\rstage}[0]{\mbox{$r$-stage}}

\newcommand{\comGNAR}[0]{\mbox{community-$\alpha$}}

\begin{document}

\title{\bf Modelling clusters in network time series with an application to presidential elections in the USA }

\author{Guy Nason\footnotemark[1] \footnotemark[2],
    Daniel Salnikov\thanks{Dept.\ Mathematics, Huxley Building, Imperial College, 180 Queen's Gate, South Kensington, London, SW7 2AZ, UK.} \footnotemark[2]\\
    Imperial College London\\
    and \\
    Mario Cortina-Borja\thanks{Great Ormond Street Institute of Child Health, 30 Guilford Street, London WC1N 1EH} \footnotemark[1] \\
    Great Ormond Street Institute of Child Health,\\
    University College London}

\date{January 15, 2024}

\maketitle

\begin{abstract}
    Network time series are becoming increasingly relevant in the study of dynamic processes characterised by a known or inferred underlying network structure. Generalised Network Autoregressive (GNAR) models provide a parsimonious framework for exploiting the underlying network, even in the high-dimensional setting. We extend the GNAR framework by introducing the \mbox{\textit{community}-$\alpha$} GNAR model that exploits prior knowledge and/or exogenous variables for identifying and modelling dynamic interactions across communities in the network. We further analyse the dynamics of {\em Red, Blue} and {\em Swing} states throughout presidential elections in the USA. Our analysis suggests interesting global and communal effects.
\end{abstract}

\noindent%

\input{introduction}

\input{gnar_communal}

\input{model_selection}

\input{application}

\input{conclusion}

\section*{Acknowledgments}
We gratefully acknowledge the following support: Nason from EPSRC NeST Programme grant EP/X002195/1; 
Salnikov from the UCL Great Ormond Street Institute of Child Health, NeST, Imperial College London, 
the Great Ormond Street Hospital DRIVE Informatics Programme and the Bank of Mexico.
Cortina-Borja supported by the NIHR Great Ormond Street Hospital Biomedical Research Centre. The views expressed are those of the authors and not necessarily those of the EPSRC.

\bibliographystyle{agsm}
\bibliography{references}

\input{appendix}

\end{document}

%% file: introduction.tex
\section{Introduction}
\label{sec: introduction}
Modelling dynamics present in network time series necessitates studying a constant flux of temporal data characterised by large numbers of interacting variables, which are associated to a network structure, e.g., networks in climate science, cyber-security, biology and political science to name a few. Traditional models, such as vector autoregressive processes (VAR), become increasingly difficult to estimate and interpret as the number of variables increases, i.e., the well known {\em curse of dimensionality}. Recently, the generalised network autoregressive (GNAR) model has been developed \cite{gnar_org, Zhu17, gnar_paper}, which provides a parsimonious model that is more interpretable {\em and} has shown superior forecasting performance in a number of settings, including the high-dimensional one, e.g., see \cite{corbit_paper}. Developments in this area include \cite{ArmillottaFokianos21, liu2023new} for Poisson/count data processes, \cite{NasonWei22} to admit time-changing covariate variables and \cite{Mantziou23, malinovskaya2023statistical} for GNAR processes on the edges of networks. We introduce the \mbox{community-$\alpha$} GNAR specification for modelling dynamic clusters in network time series; see \cite{CHEN20231239, zhu2023simultaneous} for related work. This model should be seen as an addition to the existing toolbox, rather than as a general method, which assumes prior knowledge of the network structure. Hence, it is useful when data are effectively described by an underlying network in which community structure is identifiable. Thus, it can be combined with methods that estimate network structures and/or clusters in dynamic settings. These are of interest in network and spatio-temporal modelling, e.g., \cite{clust_functional} propose methods for identifying clusters in temporal settings, and \cite{parsimonious_conf} for networked data. These can aid in identifying the communities in a network time series that can be modelled as a \mbox{community-$\alpha$} GNAR model. 

\subsection{Review of GNAR models}
\label{subsec: review of GNAR models}
A network time series $\mathcal{X} := (\boldsymbol{X}_t, \mathcal{G})$ is a stochastic process that manages interactions between nodal time series $X_{i, t} \in \mathbb{R}$ based on the underlying network $\mathcal{G}$. It is composed of a multivariate time series $\boldsymbol{X}_t \in \mathbb{R}^d$ and an underlying network $\mathcal{G} = (\mathcal{K}, \mathcal{E})$, where $\mathcal{K} = \{1, \dots, d\}$ is the node set, $\mathcal{E} \subseteq \mathcal{K} \times \mathcal{K}$ is the edge set, and $\mathcal{G}$ is an undirected graph with $d \in \mathbb{Z}^+$ nodes. Each nodal time series $X_{i, t}$ is linked to node $i \in \mathcal{K}$. Throughout this work we assume that the network is static, however, GNAR processes can handle time-varying networks; see \cite{gnar_paper}. GNAR models provide a parsimonious framework by exploiting the network structure. This is done by sharing information across nodes in the network, which allows us to estimate fewer parameters in a more efficient manner. A key notion is that of \mbox{$r$-stage} neighbours, we say that nodes $i$ and $j$ are \mbox{$r$-stage} neighbours if and only if the shortest path between them in $\mathcal{G}$ has a distance of $r$, i.e., $d(i, j) = r$. We use \mbox{$r$-stage} adjacency to define the $d \times d$ \mbox{$r$-stage} adjacency matrices $\mathbf{S}_r$, where $[\mathbf{S}_r]_{ij} := \mathbb{I} \{ d(i, j) = r \}$, $\mathbb{I}$ is the indicator function, $r \in \{1, \dots, r_{\text{max}} \}$, and $r_{\text{max}} \in \mathbb{Z}^+$ is the longest shortest path in $\mathcal{G}$. The $\mathbf{S}_r$ extend the notion of adjacency from an edge between nodes to the length of shortest paths (i.e., smallest number of edges between nodes). Note that $\mathbf{S}_1$ is the adjacency matrix and that all the $\mathbf{S}_r$ are symmetric. Further, assume that unique association weights $w_{ij} \in [0, 1]$ between nodes are available. These weights measure the relevance node $j$ has for forecasting $i$, and can be interpreted as the proportion of the neighbourhood effect attributable to node $j$. We define the weights matrix $\mathbf{W} \in \mathbb{R}^{d \times d}$ as the matrix $[\mathbf{W}]_{ij} := w_{ij}$. Note that since there are no self-loops in $\mathcal{G}$ all diagonal entries in $\mathbf{W}$ are equal to zero, and that since $w_{ij} \neq w_{ji}$ is valid $\mathbf{W}$ is not necessarily symmetric (i.e., nodes can have different degrees of relevance). 
\par
In the absence of prior weights GNAR assigns equal importance to each \mbox{$r$-stage} neighbour in a neighbourhood regression, i.e., $w_{ij} = \{| \mathcal{N}_r (i) |\}^{-1}$, where $|\mathcal{N}_r (i)| \leq d - 1$ is the number of \mbox{$r$-stage} neighbours of node $i$ and $\mathcal{N}_r(i) \subset \mathcal{K}$ is the set of \mbox{$r$-stage} neighbours of node $i$. A GNAR model assumes that effects are shared among \mbox{$r$-stage} neighbours, rather than considering pair-wise regressions, it focuses on the joint effect \mbox{$r$-stage} neighbours have on $X_{i, t}$. To do this we express the autoregressive model in terms of \mbox{$r$-stage} neighbourhood regressions. These are given by $\boldsymbol{Z}_{t}^r := \left ( \mathbf{W} \odot \mathbf{S}_r \right ) \boldsymbol{X}_t$, where $\odot$ denotes the Hadamard (component-wise) product. Each entry $Z_{i,t}^r$ in $\boldsymbol{Z}_t^r$ is the \mbox{$r$-stage} neighbourhood regression corresponding to node $i$. The vector-wise representation of a \mbox{\textit{global}-$\alpha$} $\mathrm{GNAR} \left (p, [s_k] \right )$ model is given by
\begin{equation} \label{eq: global alpha}
    \boldsymbol{X}_{t} = \sum_{k = 1}^{p} ( \alpha_k \boldsymbol{X}_{t - k} + \sum_{r = 1}^{s_k} \beta_{kr} \boldsymbol{Z}_{r, t - k} ) 
    + \boldsymbol{u}_{t}, 
\end{equation}
where $\alpha_k \in \mathbb{R}$ and $\beta_{kr} \in \mathbb{R}$ are the autoregressive coefficients, $p \in \mathbb{Z}^+$ is the maximum lag, $s_k \in \{1, \dots, r^*\}$ is the maximum \mbox{$r$-stage} depth at lag $k = 1, \dots, p$, $r^*  \leq r_{\text{max}}$ is the maximum \mbox{$r$-stage} depth across all lags, and $\boldsymbol{u}_t$ are independent and identically distributed zero-mean white noise with covariance matrix $\sigma^2_{\boldsymbol{u}} \mathbf{I_d}$ and $\sigma^2_{\boldsymbol{u}} > 0$. This compact representation is identical to the one in \cite{gnar_paper} and highlights the \textit{parsimonious} structure of a \text{global}-$\alpha$ GNAR model. The construction above follows the one in \cite{corbit_paper}, which includes more details, interpretation and further results.

%% file: gnar_communal.tex
\section{The \texorpdfstring{\mbox{community-$\alpha$}}{} GNAR model}
\label{sec: gnar methods}
Suppose that there is a collection of covariates $c \in \{1, \dots, C \} = [C]$ such that each $X_{i, t}$ is linked to only one covariate at all times $t \in \mathbb{Z}^+_0$, where $C \in \mathcal{K}$ is the number of covariates. Define $K_c := \left \{ i \in \mathcal{K} : X_{i, t} \text{ is characterised by covariate } c \right \} $. Note that by definition the $K_c$ are disjoint subsets of the node set (i.e., $K_c \subseteq \mathcal{K}$ and $K_c \cap K_{\tilde{c}} = \varnothing$ if $c \neq \tilde{c}$), and $\cup_{c =1}^{C} K_c = \mathcal{K}$. Thus, the $K_c$ form a partition of $\mathcal{K}$ and define non-overlapping clusters in $\mathcal{G}$. Intuitively, each covariate is a label that indicates the cluster to which $X_{i, t}$ belongs, e.g., if $\mathcal{G}$ consists of population centres, then each $X_{i, t}$ could be characterised as either urban, rural or a hub-town, i.e., each cluster is a collection of nodes that defines a community in $\mathcal{G}$. 
\par
The \mbox{community-$\alpha$} GNAR model is an additive model of community-wise autoregressive terms, which are obtained by using the vectors $\boldsymbol{\xi}_c \in \mathbb{R}^d$, where $\boldsymbol{\xi}_c := (\xi_{1, c}, \dots, \xi_{d, c} )$, and $\xi_{i, c} := \mathbb{I} (i \in K_c)$. Each entry in $\boldsymbol{\xi}_c$ is non-zero if and only if $i \in K_c$. The autoregressive terms are $\boldsymbol{X}_t^c := \boldsymbol{\xi}_c \odot \boldsymbol{X}_t$, note that each entry in $\boldsymbol{X}_t^c$ is not constantly zero if and only if $i \in K_c$ (i.e., $X_{i, t}$ is characterised by $c \in [C]$). Further, within community terms are given by 
$$\boldsymbol{Z}_{t - k}^{r,c} = \xi_c \odot \left ( \mathbf{W} \odot \mathbf{S}_r \right) \boldsymbol{X}_{t - k}^{c},$$ 
i.e., $r$-stage neighbourhood regressions constrained to community $K_c$. The model is given by
\begin{equation} \label{eq: community GNAR structural representation}
    \boldsymbol{X}_t = \sum_{c = 1}^{C} \left \{ \boldsymbol{\alpha}_{c} (\boldsymbol{X}_t) + \boldsymbol{\beta}_{c} (\boldsymbol{X}_t) \right \} + \boldsymbol{u}_t,
\end{equation}
where $\boldsymbol{\alpha}_{c} (\boldsymbol{X}_t) :=  \sum_{k = 1}^{p_c} \alpha_{k, c} \boldsymbol{X}_{t - k}^{c}$ is the community autoregressive component, and 
\newline
$\boldsymbol{\beta}_{c} (\boldsymbol{X}_t) := \sum_{k = 1}^{p_c} \sum_{r = 1}^{s_k (c)} \beta_{k, r, c} \boldsymbol{Z}_{t - k}^{r, c}$ is the within community component for each community $K_c$. Above in \eqref{eq: community GNAR structural representation}, $\alpha_{k, c} \in \mathbb{R}$ are autoregressive coefficients at lag $k$ for $K_c$, $\beta_{k, r, c} \in \mathbb{R}$ are \mbox{$r$-stage} neighbourhood regression coefficients at lag $k$ for $K_c$, and $\boldsymbol{u}_t$ are zero-mean independent and identically distributed white noise such that $\mathrm{cov} ( \boldsymbol{u}_t) = \sigma^2_{\boldsymbol{u}} \mathbf{I_d}$ and $\sigma^2_{\boldsymbol{u}} > 0$. We denote the model order of \eqref{eq: community GNAR structural representation} by community-$\alpha$ $\mathrm{GNAR} ([p_c], \{ [s_k (c)] \}, [C])$, where $p_c \in \mathbb{Z}^+$ is maximum lag and $s_k (c) \leq r_{\max}$ is maximum $r$-stage at lag $k$ for $K_c$, $k = 1, \dots, p$ is current lag, $p = \max (p_c)$ is global maximum lag, $C$ is the number of communities, and $c \in [C]$ is the covariate that characterises community $K_c$. The model given by \eqref{eq: community GNAR structural representation} is stationary if its parameters satisfy 
$\sum_{k = 1}^{p_c} \left \{ |\alpha_{k, c} | + \sum_{r = 1}^{s_k (c)} |\beta_{k, r, c} |  \right \} < 1,$
for all covariates $c \in [C]$. This is a direct application of results in \cite{gnar_paper}.

\begin{remark} \label{rem: VAR representation}
    Expressing \eqref{eq: community GNAR structural representation} as a {\em VAR} is done by incorporating networked-informed constraints into autoregressive matrices. Let $\mathbf{\Phi}_{k}$ be the $d \times d$ matrix given by
    \begin{equation} \label{eq: var autoregressive matrices}
        \mathbf{\Phi}_{k} = \sum_{c = 1}^C \left [ \mathrm{diag} ( \alpha_{k, c} \boldsymbol{\xi}_c ) + \sum_{r = 1}^{s_k (c)} \left \{ \beta_{k, r, c} (\mathbf{W}_c \odot \mathbf{S}_r ) \right \} \right ],
    \end{equation}
    where terms for larger order are set to zero, e.g., if $p_c < p_{\tilde{c}}$, then $\alpha_{k, c} \equiv 0$ for $k > p_c$, $[\mathbf{W}_c]_{ij} = w_{ij} \mathbb{I} ( i \in K_c \text{ and } j \in K_c)$, i.e., $\mathbf{W}$ constrained to community $K_c$. Then, the {\em VAR}$(p)$ model given by
    $ \boldsymbol{X}_t = \sum_{k = 1}^p \mathbf{\Phi}_{k} \boldsymbol{X}_{t - k} + \boldsymbol{u}_t,$
    where $\boldsymbol{u}_t$ are {\em i.i.d.} white noise, is identical to the model given by \eqref{eq: community GNAR structural representation}.
\end{remark}

\subsection{Model estimation}
Estimation of GNAR models is straightforward by noting that these are network-informed constrained VAR models; see \cite{gnar_paper, corbit_paper}. However, we present a conditional linear model that exhibits the parsimonious nature of GNAR processes and aids interpretation. Assume that we observe $T \in \mathbb{Z}^+$ time-steps of a stationary \mbox{community-$\alpha$} GNAR process with known order. The data $\mathbf{X} := [\boldsymbol{X}_1, \dots, \boldsymbol{X}_T ]$ are a realisation of length $T$ coming from a stationary $\mathrm{GNAR} ([p_c], \{ [s_k (c)] \}, [C])$. Notice that we can concatenate each community term in \eqref{eq: community GNAR structural representation} into design matrices as follows
\begin{align} \label{eq: time-step desing matrix}
    \mathbf{R}_{k, t, c} &:= \left [\boldsymbol{X}_{t - k}^c |  \boldsymbol{Z}_{t - k}^{1, c} | \dots |  \boldsymbol{Z}_{t - k}^{s_k (c), c} \right], \nonumber \\
    \mathbf{R}_{k, t} &:= \left [ \mathbf{R}_{k, t, 1}| \dots | \mathbf{R}_{k, t, C} \right ], \\
    \mathbf{R}_{t} &:= \left [ \mathbf{R}_{1, t}| \dots | \mathbf{R}_{p, t} \right ], \nonumber
\end{align}
where predictor columns are concatenated in ascending order with respect to $c$, i.e., if $\tilde{c} > c$, then the columns for $c$ precede the ones for $\tilde{c}$. Hence, $ \mathbf{R}_{k, t, c}$ is the design matrix for $K_c$ at lag $k$, $ \mathbf{R}_{k, t}$ is the design matrix for all communities at lag $k$ and $ \mathbf{R}_{t}$ is the design matrix for all communities and lags $k = 1, \dots, p$. 
By stacking the $ \mathbf{R}_{t}$ for $t = p + 1, \dots, T$, and defining $\boldsymbol{\theta} := (\boldsymbol{\theta}_1, \dots, \boldsymbol{\theta}_C ) \in \mathbb{R}^{q}$, where $\boldsymbol{\theta}_c = (\alpha_{1, c}, \beta_{1, 1, c}, \dots, \beta_{1, s_1(c), c}, \alpha_{2, c}, \dots, \beta_{p_c, s_{p}(c), c}) \in \mathbb{R}^{q_c}$, is ordered by lags (i.e., all parameters are stacked for each lag), is the vector of parameters for $K_c$, and $q = \sum_{c = 1}^C q_c$ is the number of unknown parameters. We can write \eqref{eq: community GNAR structural representation} as the linear model 
\begin{equation} \label {eq: linear model}
    \boldsymbol{y} = \mathbf{R} \boldsymbol{\theta} + \boldsymbol{u},
\end{equation}
where $\boldsymbol{y} = (\boldsymbol{X}_{p + 1}, \dots, \boldsymbol{X}_{T}) \in \mathbb{R}^{d(T - p)}$ is the response, $\mathbf{R}$ is the $d(T - p) \times q$ design matrix, and entries in $\boldsymbol{u} = (\boldsymbol{u}_{p + 1}, \dots, \boldsymbol{u}_{T})$ are {\em i.i.d.} white noise. Thus, we can estimate $\boldsymbol{\theta}$ by least-squares, i.e., we use the estimator $\hat{\boldsymbol{\theta}} = \left ( \mathbf{R}^T \mathbf{R} \right )^{-1} \mathbf{R}^T \boldsymbol{y}$ throughout this work.
\begin{remark} \label{rem: estimators}
    Assume that $\boldsymbol{X}_t$ is a stationary {\em community}-$\alpha$ {\em GNAR} model with {\em i.d.d.} white noise residuals, and note that each $\mathbf{R}_c$ has zeros in different rows (non-overlapping communities). Then, $\mathbf{R}$ is orthogonal by blocks, and since $\mathrm{cov} (\hat{\boldsymbol{\theta}} ) = \sigma^2_{\boldsymbol{u}} (\mathbf{R}^T \mathbf{R})^{-1}$, the $\hat{\boldsymbol{\theta}}_c$ are uncorrelated and non-zero entries in the precision matrix, $\{ \mathrm{cov} (\hat{\boldsymbol{\theta}} ) \}^{-1}$, correspond to estimated coefficients in the same community (e.g., $\mathrm{cov} (\hat{\alpha}_{k, c}, \hat{\alpha}_{k, \tilde{c}}) = 0$ if $c \neq \tilde{c}$). Further, if we assume that $\boldsymbol{u}_t \sim \mathrm{N}_{\boldsymbol{d}} (\boldsymbol{0}, \sigma^2_{\boldsymbol{u}} \mathbf{I_d} )$, then $\hat{\boldsymbol{\theta}}$ is the conditional maximum likelihood estimator and $\hat{\boldsymbol{\theta}}_c = \left ( \mathbf{R}_c^T \mathbf{R}_c \right )^{-1} \mathbf{R}_c^T \boldsymbol{y}$ are block-wise independent.
\end{remark}

Note that by Remark \ref{rem: estimators}, it is possible to estimate model parameters separately and simultaneously for community-$\alpha$ GNAR models. This allows us to use more observations for communities with a smaller maximum lag, remove unnecessary predictors from each $c$-community linear model, and perform estimation in parallel, which is useful for very large networks with a lot of observations, e.g., internet traffic network time series. Further, adapting $\hat{\boldsymbol{\theta}}$ to a generalised least-squares setting is straightforward. Suppose that $\mathrm{cov} (\boldsymbol{u}) = \mathbf{\Sigma}_T$. Then, we can estimate $\boldsymbol{\theta}$ by generalised least-squares, i.e., 
\begin{equation} \label{eq: gls estimator}
\hat{\boldsymbol{\theta}}_{\mathrm{gls}} = \left ( \mathbf{R}^T \mathbf{\Sigma}_T^{-1} \mathbf{R} \right )^{-1} \mathbf{R}^T  \mathbf{\Sigma}_T^{-1} \boldsymbol{y},    
\end{equation}
where $\mathbf{\Sigma}_T$ is a valid $d(T - p) \times d(T - p)$ covariance matrix, e.g., $ \mathbf{\Sigma}_T = \mathbf{I_d} \otimes \mathbf{\Sigma}_{\boldsymbol{u}}$, where $\otimes$ denotes Kronecker product, which is block-diagonal with entries $\mathrm{cov}(\boldsymbol{u}_t) = \mathbf{\Sigma}_{\boldsymbol{u}} $ at all times $t$. Moreover, the linear model in \eqref{eq: linear model} can be broken into its community components, which can be estimated by different strategies, e.g., some communities could be regularised and/or estimated using more robust estimators. Also, it is possible to express dependence between residuals in a community-wise manner, i.e., assuming that $\mathbf{\Sigma}_{\boldsymbol{u}}$ is block diagonal, where each block corresponds to one community.

%% file: model_selection.tex
\section{Model Selection}
\label{sec: model selection}
We perform model selection by examining the graphical aids suggested by \cite{corbit_paper}. These are the Corbit (correlation-orbit) and R-Corbit plot, which aid us in visualising the network autocorrelation function (NACF) and the partial NACF (PNACF). These are network-enabled extensions of the autocorrelation and partial autocorrelation functions from univariate time series analysis. The R-Corbit plot is particularly relevant for our purposes of detecting whether or not the order differs across communities. Each point in a Corbit plot corresponds to either $\nacf(h, r)$ or $\pnacf(h, r)$, where $h$ is the $h$th lag and $r$ is $\rstage$ depth. The numbers on the outermost ring indicate lag, and $\rstage$ depth is read by ring order starting from the inside (i.e., the innermost ring is for $r = 1$, the second one for $r = 2$, etc ...). Points in a R-Corbit plot correspond to $\mathrm{(p)nacf}_c (h, r)$, where $h$ is the $h$th lag, $r$ is $\rstage$ depth and $c \in \{1, \dots, C\}$ is the community. These belong to a circle ring, where the mean value, i.e., $C^{-1} \sum_{c = 1}^{C} \mathrm{(p)nacf}_c (h, r)$, is shown at the centre. The numbers on the outside ring indicate lags and each inner ring of circle rings indicates $\rstage$ depth starting from one. Both plots show (P)NACF equal to zero (i.e., $\nacf = 0$ or $\pnacf = 0$) at the centre to aid comparison. See \cite{corbit_paper} for the NACF and PNACF definitions and more details.
\subsection{GNAR simulation}
We simulate ten realisations of length one-hundred coming from a stationary $\comGNAR$ GNAR model given by \eqref{eq: community GNAR structural representation}. Each simulation is prodced using the \textbf{fiveNet} network (included in the \texttt{GNAR} package) as the underlying network; see \cite{gnar_package} and Figure \ref{fig: 2-communal fiveNet}. In what follows, we use the \texttt{GNAR} package for producing Corbit and R-Corbit plots.
\begin{figure}
   \centering
   \includegraphics[scale=0.40]{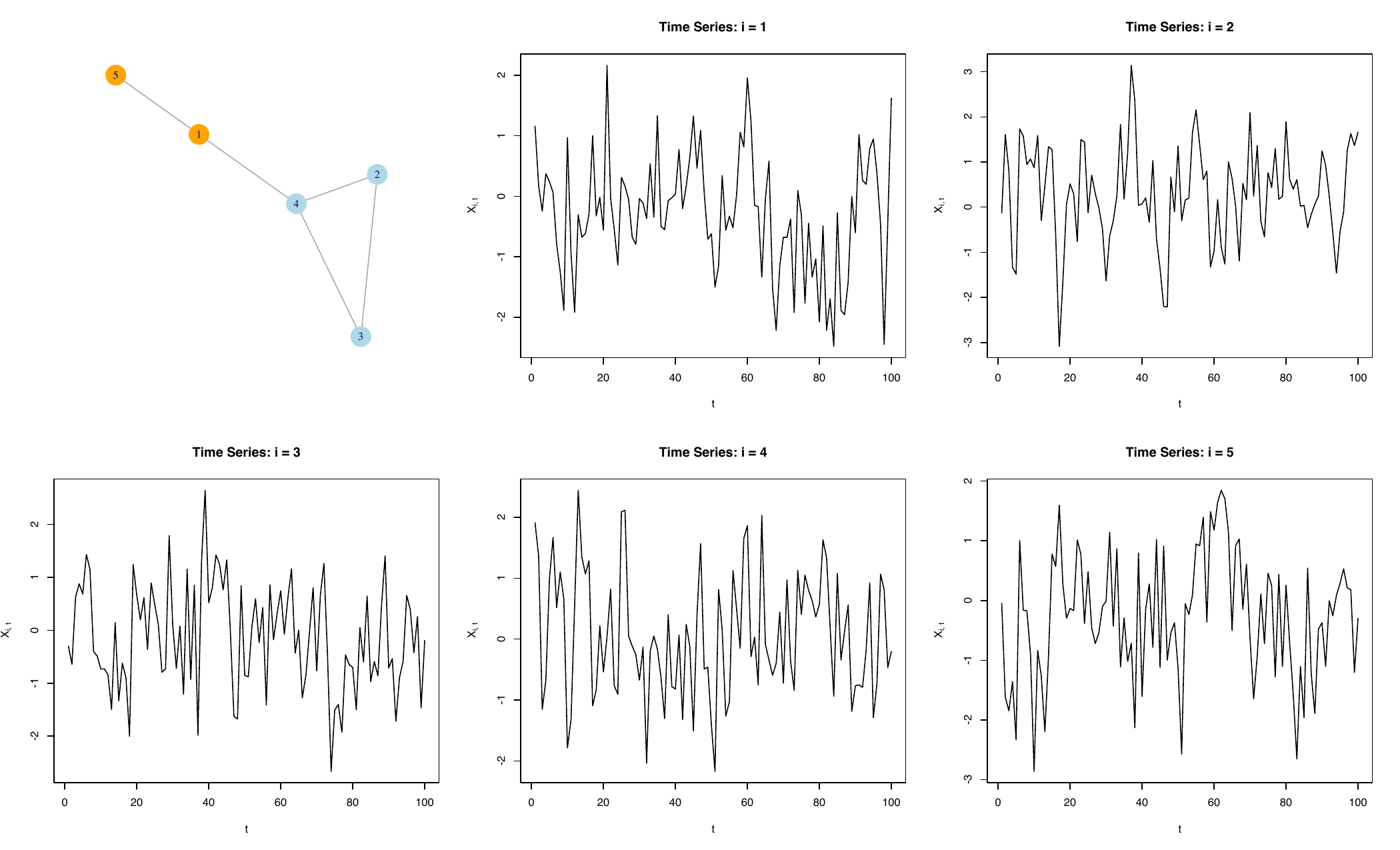}
   \caption{
       The first plot is the \textbf{fiveNet} network (included in the \texttt{GNAR} package); see \cite{gnar_package}. The latter are nodal time series realisations of length 100 coming from a stationary $\comGNAR$ $\GNAR \left ( [1, 2], \{ [1], [1, 1] \}, [2] \right )$, where \textbf{fiveNet} is the underlying network. Blue nodes correspond to community $K_1 = \{2, 3, 4 \}$ and orange ones to community $K_2 = \{1, 5 \}$.
   }
   \label{fig: 2-communal fiveNet}
\end{figure}
The data come from a stationary $\GNAR \left ( [1, 2], \{ [1], [1, 1] \}, [2] \right )$, where $K_1 = \{2, 3, 4\}$ and $K_2 = \{1, 5 \}$ are the two communities. Suppose we no longer know the model order but still know the communities. We proceed to study lag and $\rstage$ order, i.e., the pair $(p_c, [s_{k_c}] )$, by comparing the community correlation structure via the R-Corbit plot in Figure \ref{fig: R-Corbit pnacf} below. 
\begin{figure}
   \centering
   \includegraphics[scale=0.40]{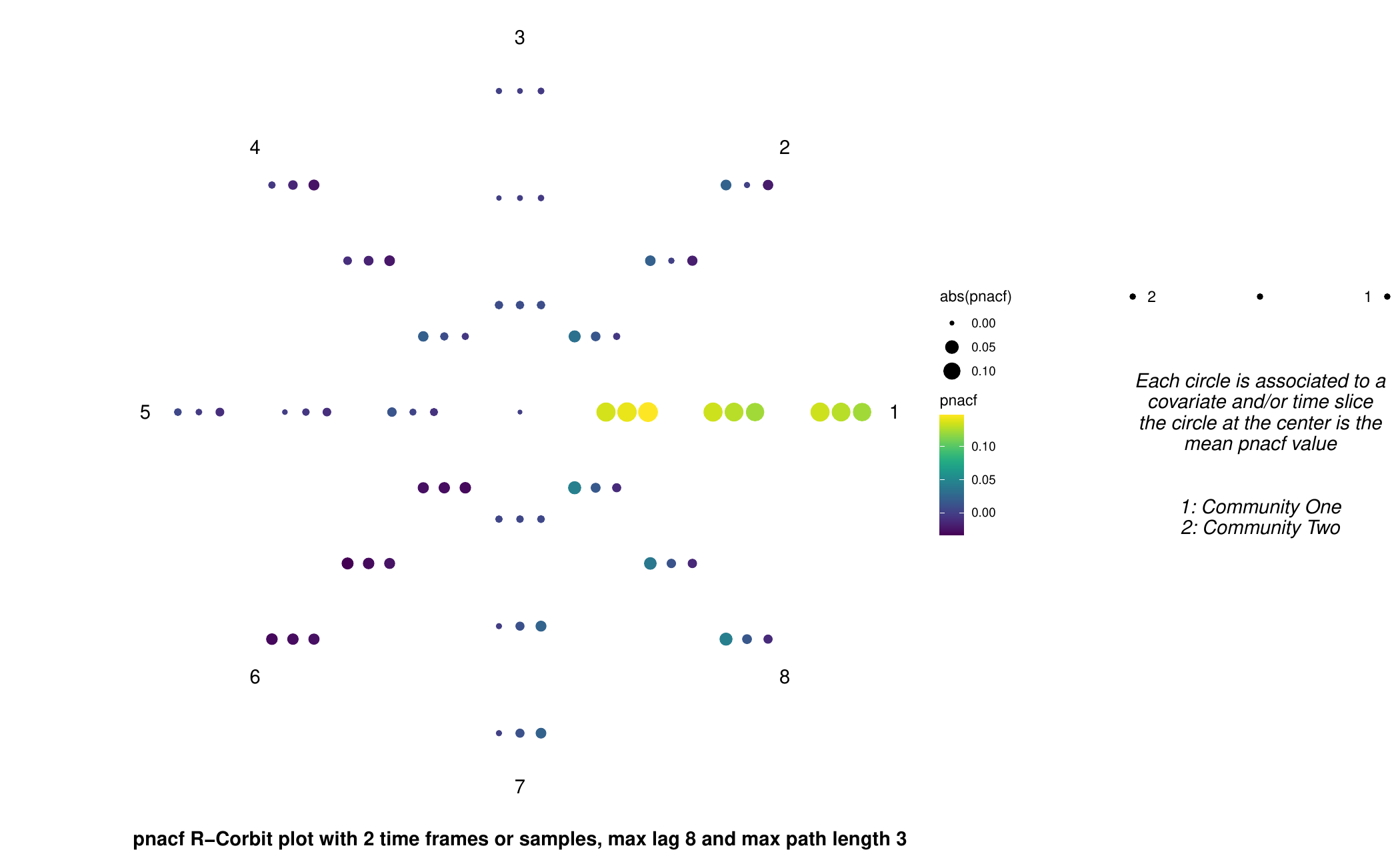}
   \caption{
       R-Corbit plot for a realisation of length 100 coming from a stationary $\comGNAR$ $\GNAR \left ( [1, 2], \{ [1], [1, 1] \}, [2] \right )$, where the underlying network is \textbf{fiveNet}, $K_1 = \{2, 3, 4\}$ and $K_2 = \{1, 5 \}$; see Figure \ref{fig: 2-communal fiveNet}. The maximum lag is equal to eight and maximum $\rstage$ depth is equal to three. The PNACF cut-offs are $(1, [1])$ for $K_1$ and $(2, [1, 1])$ for $K_2$.
        }
   \label{fig: R-Corbit pnacf}
\end{figure}
The R-Corbit plot in Figure \ref{fig: R-Corbit pnacf} shows that the PNACF cuts-off after the second lag for both communities. Furthermore, it shows that the PNACF cuts-off after the first lag and $\rstage$ one for $K_1$, i.e., it suggests the order $(1, [1])$ for $K_1$, and after the second lag at the first $\rstage$ for lags one and two for $K_2$, i.e., it suggests the pair $(2, [1, 1])$ for $K_2$. This is common for all ten R-Corbit plots and reflects the known data-generating process. Table \ref{tab: sim fitted model} shows the estimated coefficients for our chosen model $\GNAR  \left ( [1, 2], \{ [1], [1, 1] \}, [2] \right )$.

\begin{table}[H]
   \centering
   \begin{tabular}{cccc}
         $\hat{\boldsymbol{\theta}}$ & Mean Est. & Sim. Sd. & True Value \\
        \hline
        $\hat{\alpha}_{11}$ & 0.28 & 0.062 & 0.27 \\
        $\hat{\beta}_{111}$ & 0.177 & 0.109 & 0.18 \\
        $\hat{\alpha}_{12}$ & 0.222 & 0.061 & 0.25 \\
        $\hat{\beta}_{112}$ & 0.302 & 0.098 & 0.30 \\
        $\hat{\alpha}_{22}$ & 0.115 & 0.083 & 0.12 \\
        $\hat{\beta}_{212}$ & 0.18 & 0.127 & 0.20
   \end{tabular}
   \caption{
    Mean estimator values for 10 simulations, each of length 100 coming from a $\GNAR  \left ( [1, 2], \{ [1], [1, 1] \}, [2] \right )$. Mean Est. is the mean estimator value. Sim. Sd is the standard deviation of the 10 simulations. True Value is the known parameter. The simulations are computed using the random seeds 1983 to 1993 in \texttt{R}.
   }
   \label{tab: sim fitted model}
\end{table}
Table \ref{tab: sim fitted model} shows the reasonable performance of our estimator $\hat{\boldsymbol{\theta}}$. We end this subsection by noting that Corbit and R-Corbit plots enable quick identification of model order, different correlation structures between communities, and other behaviours such as seasonality and trend for (network) time series. Moreover, the R-Corbit plot in Figure \ref{fig: R-Corbit pnacf} suggests that a test of significance could be performed using the PNACF to identify model order and/or strong correlations, which might point to network effects. We leave these investigations for future work.

%% file: application.tex
\section{Modelling presidential elections in the USA}
\label{sec: application}
The data for our study are obtained from the MIT Election Data and Science Lab and can be accessed at 
({\tt doi.org/10.7910/DVN/42MVDX}). We study the twelve presidential elections in the USA from 1976 to 2020. Denote this network time series by $(\boldsymbol{X}_t, \mathcal{G})$, where $X_{i, t}$ is the percentage of votes for the Republican nominee in the $i$th state (ordered alphabetically) for election year $t \in \{1976, 1980, \dots, 2020 \}$. The network is $\mathcal{G} = (\mathcal{K}, \mathcal{E})$, where $i \in [51]$ and $d = 51$, and it is built by connecting states that share a land border (i.e., there is an edge between two nodes if and only if their respective states share a land border). Following \cite{state_classification_econ_enquiry, state_classification_science_journal}, we classify each node (state) as either \textit{Red}, \textit{Blue} or \textit{Swing} based on the percentage of elections won by either party. The communities are: $i \in K_1$ if the Republican nominee won at least $75\%$ of elections, $i \in K_2$ if the Democrat nominee won at least $75\%$ of elections, and $i \in K_3$ if neither nominee won at least $75\%$ of elections; see Figure \ref{fig: usa net}. 
\par
In what follows, we use the CRAN \texttt{GNAR} package for computing the network autocorrelation function (NACF) and partial NACF (PNACF), and producing R-Corbit plots. The R-Corbit plot in Figure \ref{fig: R-corbit pnacf votes} shows that the PNACF is positive at the first lag, negative and strongest at the second lag, cuts-off at lags three and four, and, interestingly, appears to be strong at the fifth lag across all $r$-stages. At the first lag, the PNACF cuts-off after the first $r$-stage, and at both the second and fifth lags decays as $r$-stage grows but does not cut-off at any $r$-stage. This suggests a positive correlation for elections in which a president is running for reelection, and that network effects influence said election. Remarkably, the strong correlation at the second lag across all $r$-stages suggests a change in the system, which we interpret as alternating between Republican and Democrat nominees once the incumbent president has completed their eight-year term. This has been the case with the exceptions of Jimmy Carter (1976-1980), George Bush (1988-1992) and Donald J. Trump (2016-2020). Interestingly, the exception cases are the ones in which there was a change at the election in which an incumbent president was running for reelection. We believe that the fifth lag might be identifying these oddities. Nevertheless, more analysis is needed.

\begin{table}[H]
    \centering
     \begin{tabular}{ccccccc}
        & GNAR & GNAR* & GNAR+ & sp. VAR & CARar & Naive\\
        \hline
         rMSPE & 3.18 & 2.34 & 24.95 & 72.57 & 3.60 & 2.45 \\
        rMSPE* & 7.29 & 3.00 & 5.47 & 4.86 & 2.75 & 2.45 \\
        $\# ||\boldsymbol{\theta}||$ & 9 & 3 & 377 & 6 & NA
    \end{tabular}
    
    \caption{Model comparisons. GNAR: community-$\alpha$. GNAR*: global-$\alpha$. GNAR+: local-$\alpha$.
    CARar$(2)$: spatio-temporal conditional autoregressive.  
    sp. VAR: Sparse VAR$(2)$. Naive: previous observation. rMSPE $=\{ \sum_{i = 1}^{51} (X_{i, t} - \hat{X}_{i, t})^2 / 51 \}^{1/2} $. rMSPE* is for mean-centred data, i.e., $Y_{i, t} = X_{i, t} - \overline{X}_{i, t}$, $\overline{X}_{i, t} = \sum_{t = 1}^{11} X_{i, t} / 11$. $\# ||\boldsymbol{\theta}||$ is number of parameters.
    }
    \label{tab: model comparisons}
\end{table}

\begin{figure}[H]
    \centering
    \includegraphics[scale=0.40]{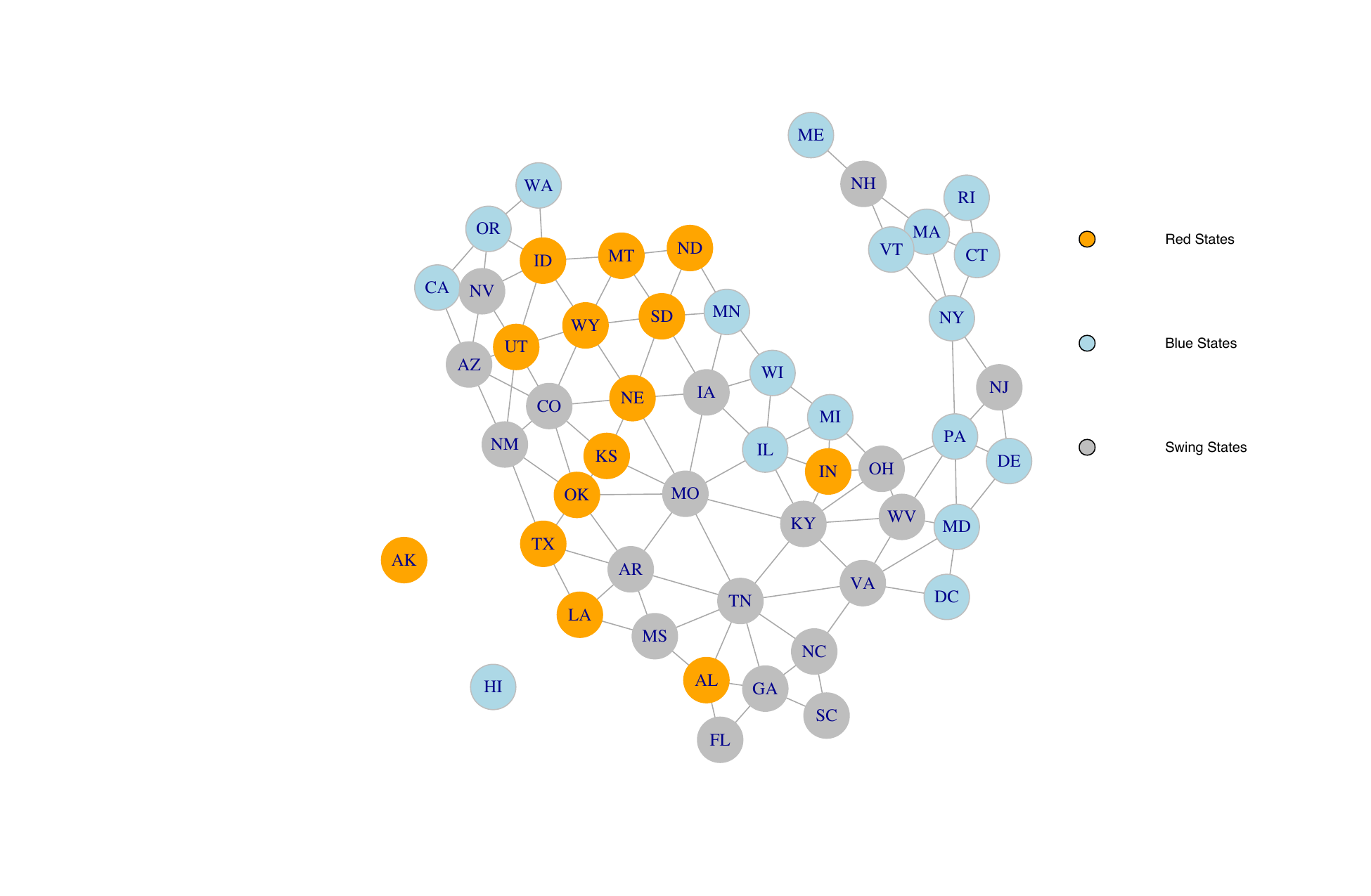}
    \caption{
    USA state-wise network, blue nodes are {\em Blue} states (Democrat Nominee won at least 75$\%$ of elections), orange nodes are {\em Red} states (Republican nominee won at least $75\%$ of elections), and grey nodes are {\em Swing} states (neither party won at least $75\%$ of elections).
    }
    \label{fig: usa net}
\end{figure}

Figure \ref{fig: R-corbit pnacf votes} suggests a model order of lag two and first stage neighbours for the three communities. Thus, our choice of model is a community-$\alpha$ GNAR$([2], \{ [1, 0]\}, [3])$. Table \ref{tab: model comparisons} compares our choice with alternative models. We fix the random seed (e.g., \texttt{set.seed(2024)}), fit sparse VAR using \texttt{sparsevar}; see \cite{sparsevar}, and CARar (forecasts are computed as a global-$\alpha$ GNAR) using \texttt{CARBayesST}; see \cite{carbayes}. Remarkably, global-$\alpha$ GNAR$(2, [1, 0])$ forecasts produce the smallest root mean squared prediction error. However, this is an atypical election (2020) and the dataset is overly sparse, thus, the models' predictive capabilities are likely limited. This preliminary analysis suggests, as expected, that {\em Red} states vote mostly for the Republican nominee, {\em Blue} states for the Democrat nominee, and that {\em Swing} states play the deciding role, however, it also suggests that the clusters behave more similarly than expected, i.e., there appear to be (spatio-temporal) global and communal network effects.

\begin{figure}
    \centering
    \includegraphics[scale=0.45]{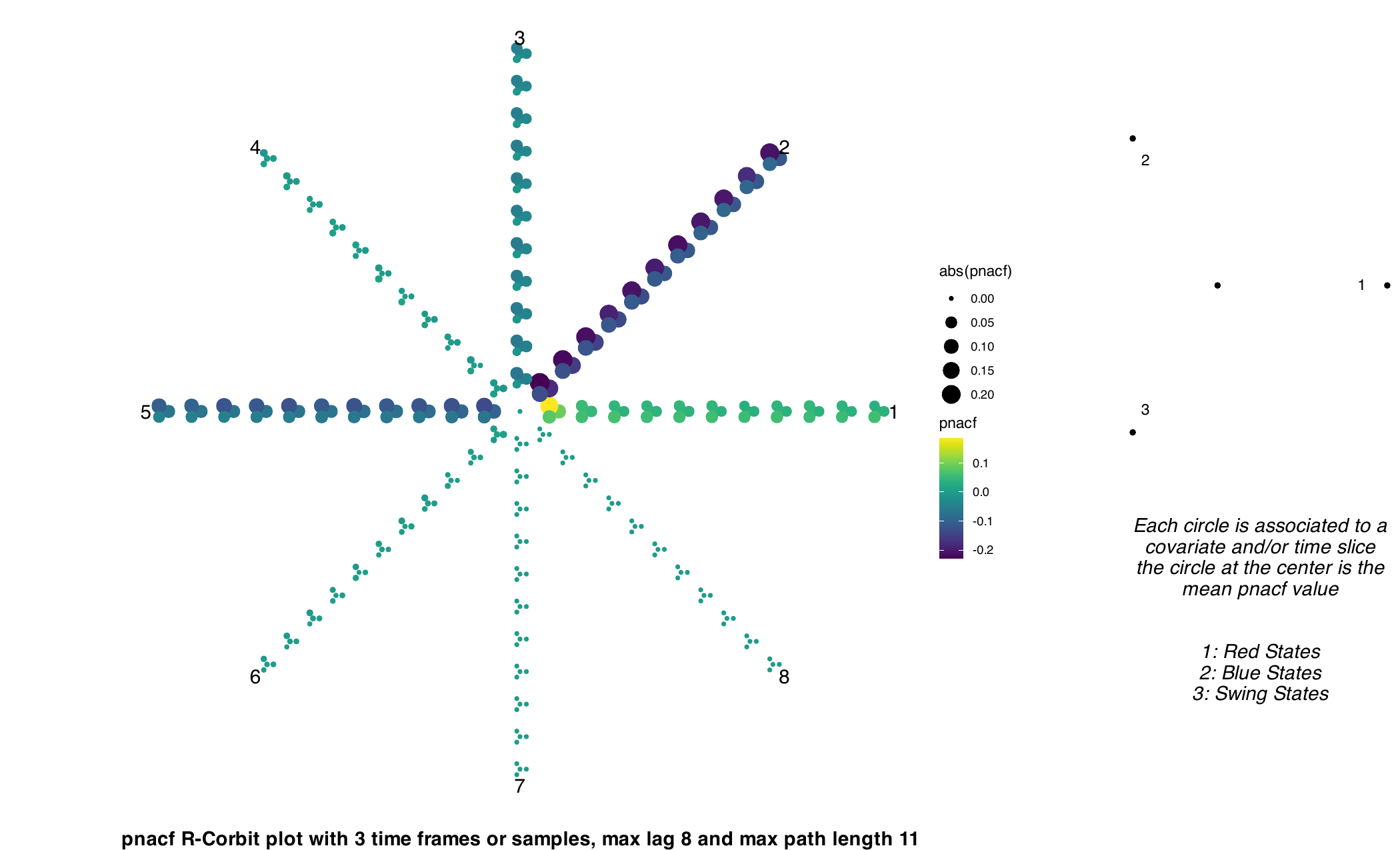}
    \caption{PNACF R-Corbit plot of presidential state-wise percentage vote for Republican nominee. Points in a R-Corbit plot correspond to $\mathrm{(p)nacf}_c (h, r)$, where $h$ is $h$th lag, $r$ is $r$-stage depth and $c \in [C]$ is the community. These belong to a circle ring, where the mean value, i.e., $C^{-1} \sum_{c = 1}^{C} \mathrm{(p)nacf}_c (h, r)$, is shown at the centre. The numbers on the outermost ring indicate lag, and $r$-stage depth is read by ring order starting from the inside (i.e., innermost ring is for $r = 1$, second one for $r = 2$, etc ...). The underlying network is the USA state-wise network; see Figure \ref{fig: usa net}.}
    \label{fig: R-corbit pnacf votes}
\end{figure}

Based on the above we select a \mbox{community-$\alpha$} $\mathrm{GNAR} (2, \{[1, 0]\}, [3])$ model. The parsimonious nature of GNAR models allows us to estimate a model in spite of having twelve observations for a vector of dimension fifty-one. Our chosen model requires at least nine observations for estimation, in contrast, a VAR(1) model for these data requires at least fifty-one observations. Thus, it is not possible to fit a VAR($p$) to these data without performing regularisation, which usually requires cross-validation and a reasonable number of observations. The results follow in Table \ref{tab: usa vote estimate} below.

\begin{table}[H]
    \begin{tabular}{cccccccccc}
        & $\hat{\alpha}_{11}$ & $\hat{\beta}_{111}$ & $\hat{\alpha}_{21}$ & $\hat{\alpha}_{12}$ & $\hat{\beta}_{112}$ & $\hat{\alpha}_{22}$ & $\hat{\alpha}_{13}$ & $\hat{\beta}_{113}$ & $\hat{\alpha}_{23}$ \\
        \hline
        Estimate & 0.393 & 0.183 & -0.593 & 0.558 & 0.069 & -0.351 & 0.905 & -0.747 & -0.591 \\
        Std. Err. & 0.104 & 0.200 & 0.066 & 0.108 & 0.178 & 0.063 & 0.089 & 0.162 & 0.061 \\
    \end{tabular}
    \caption{Parameter estimates for the standardised (i.e., $(X_{i, t} - \overline{X}_i) \{ \sum_{t = 1}^{12} (X_{i, t} - \overline{X}_{i} )^2  \}^{- 1 / 2} $ ) network time series $(\boldsymbol{X}_{t}, \mathcal{G})$ of vote percentages for the Republican nominee in presidential elections in the USA from 1976 to 2020. The fit is a \mbox{community-$\alpha$} $\mathrm{GNAR} (2, \{[1, 0]\}, [3])$. }
\label{tab: usa vote estimate}
\end{table}

The estimated coefficients in Table \ref{tab: usa vote estimate} above show the strong network effects at the first lag, the strongest being the one for {\em Swing} states. This suggests that, on average, {\em Red} and {\em Blue} states have a weaker correlation with neighbouring states than {\em Swing} states do (i.e., {\em Swing} states listen more intently). However, the estimated coefficients in Table \ref{tab: usa vote estimate} do not satisfy stationarity conditions. We circumvent this by studying the well-known one-lag differenced series. The R-Corbit plot in Figure \ref{fig: r-corbit pnacf one-lag diff} suggests that one-lag differencing removes correlation across all $r$-stages for the first lag, and that the PNACF cuts-off after lag three, and that it does not cut-off at any particular $r$-stages. We fit a community-$\alpha$ $\mathrm{GNAR}(3, \{[0, 0, 0]\}, [3])$ to the one-lag differenced data. The results follow in Table \ref{tab: usa vote one-lag diff estimate}.

\begin{table}[H]
    \begin{tabular}{cccccccccc}
        & $\hat{\alpha}_{21}$ & $\hat{\alpha}_{21}$ & $\hat{\alpha}_{31}$ & $\hat{\alpha}_{12}$ & $\hat{\alpha}_{22}$ & $\hat{\alpha}_{32}$ & $\hat{\alpha}_{13}$ & $\hat{\alpha}_{23}$ & $\hat{\alpha}_{33}$ \\ 
        \hline
        Estimate & -0.081 & -0.538 & -0.303 & -0.139 & -0.571 & -0.251 & 0.029 & -0.582 & -0.178 \\
        Std. Err. & 0.086 & 0.064 & 0.082 & 0.076 & 0.055 & 0.069 & 0.08 & 0.054 & 0.072 \\
    \end{tabular}
    \caption{Estimated coefficients for the standardised (i.e., $(X_{i, t} - \overline{X}_i) \{ \sum_{t = 1}^{11} (X_{i, t} - \overline{X}_{i} )^2  \}^{-1 / 2}) $ one-lag differenced network time series $(\boldsymbol{X}_{t}, \mathcal{G})$ of vote percentages for the Republican nominee in presidential elections in the USA from 1980 to 2020. The fit is a \mbox{community-$\alpha$} $\mathrm{GNAR} (3, \{[0, 0, 0]\}, [3])$. }
\label{tab: usa vote one-lag diff estimate}
\end{table}

Interestingly, one-lag differencing appears to remove network effects, and the estimated coefficients satisfy our stationarity assumptions. Therefore, it appears that there are community-wise trends, however, these trends appear to be dominated by the system change at the second lag (two new candidates running for president). This is the strong global effect. Thus, there appears to be global and communal effects in the network. Future work will examine the interactions between different communities and how this might help explain the relationship between global and communal effects.

\begin{figure}[H]
    \centering
    \includegraphics[scale=0.45]{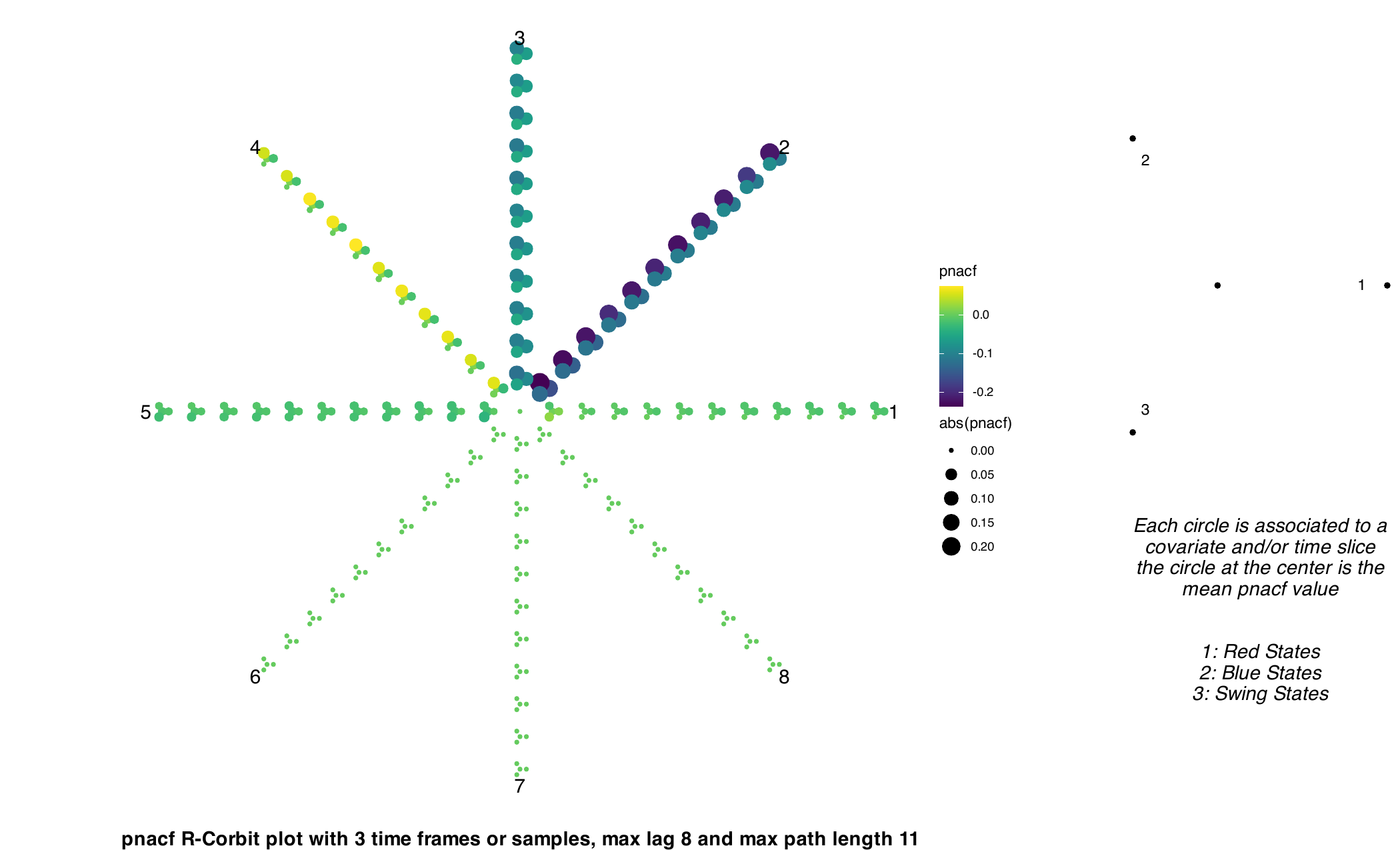}
    \caption{R-Corbit plot for the one-lag differenced network time series $\boldsymbol{X}_t - \boldsymbol{X}_{t - 1}$. Each $X_{i, t}$ is the difference in the percentage of votes for the Republican nominee in the current election minus the previous one from 1980 to 2020. The communities are {\em Red, Blue} and {\em Swing} states; see Figure \ref{fig: usa net}.}
    \label{fig: r-corbit pnacf one-lag diff}
\end{figure}

%% file: conclusion.tex
\section{Conclusion}
\label{sec: conclusion}
We have presented the \mbox{community}-$\alpha$ GNAR model, its parsimonious framework allows analysing high-dimensional (network) time series data, and can detect interesting community dynamics, see, e.g., Section \ref{sec: application}. However, it requires knowledge of network communities, and assumes stationarity. Future work will focus on extending the methodology and a more thorough analysis of the electoral data. We mention some interesting methodological developments. Note that community-wise correlation structures suggest a natural clustering algorithm. If we assume that there is a certain number number of communities in the network, then clustering a network time series could be done by assigning each node to the community that minimises some criteria based on the network autocorrelation function, which takes into account weights, network structure and communal autoregressive effects. Further, extending GNAR to nonstationary (bio)spatio-temporal settings could enable modelling of complex mixed type data. The code for replicating the study and model fitting will be added to the CRAN \texttt{GNAR} package in due course.

%% file: appendix.tex
\section*{Appendix}
\subsection*{Stationarity conditions for \texorpdfstring{\mbox{community-$\alpha$}}{} GNAR models}
We will use the following result of \cite{gnar_paper}.
\begin{theorem}{[\textbf{\cite{gnar_paper}}]} \label{th: stationary cond}
Let $\boldsymbol{X}_t$ be a \mbox{local-$\alpha$} $\GNAR (p, [s_1, \dots, s_p] )$ process with associated static network $\mathcal{G} = (\mathcal{K}, \mathcal{E})$. Further, assume that the autoregressive coefficients in 
$$ X_{i, t} = \sum_{k = 1}^{p} \left ( \alpha_{i, k} X_{i, t - k} + \sum_{r = 1}^{s_k} \beta_{kr} Z_{i, t - k}^{r} \right ) 
    + u_{i, t},$$ 
where $u_{i, t}$ are IID white noise, satisfy
$$
    \sum_{k = 1}^{p} \left ( |\alpha_{i, k}| + \sum_{r = 1}^{s_k} |\beta_{k r} | \right ) < 1
$$
for all $X_{i, t}$, $i \in \{1, \dots, d\}$. Then $\boldsymbol{X}_t$ is stationary.
\end{theorem}

We proceed to show the following. 
\begin{corollary}
Let $\boldsymbol{X}_t$ be a $\comGNAR$ $\GNAR\left ( [p_c], \{[s_{k_c}] \}, [C] \right )$ process such that the autoregressive coefficients in \eqref{eq: community GNAR structural representation} satisfy 
$$
    \sum_{k = 1}^{p_c} \left ( |\alpha_{k_c, c} | + \sum_{r = 1}^{s_{k_c}} |\beta_{k_c, r, c} | \right ) < 1
$$
for all covariates $c \in [C]$. Then $\boldsymbol{X}_t$ is stationary.
\end{corollary}

\begin{proof}
Notice that by Theorem \ref{th: stationary cond} each $c$-community \mbox{global-$\alpha$} GNAR process within curly brackets in \eqref{eq: community GNAR structural representation} on the right-hand side is stationary if the coefficients for community $K_c$ satisfy 
\begin{equation*}
    \sum_{k = 1}^{p_c} \left ( |\alpha_{k_c, c} | + \sum_{r = 1}^{s_{k_c}} |\beta_{k_c, r, c} | \right ) < 1,
\end{equation*}
where  $\alpha_{k_c, c}$, $\beta_{k_c, r, c}$ are the model coefficients, $k_c \in \{1, \dots, p_c\}$ is the current lag, $s_{k_c} \in \{1, \dots, r^*_c\}$ is the maximum $\rstage$ depth at lag $k_c$, $p_c \in \mathbb{Z}^+$ and $r^*_c$ are the maximum lag and $\rstage$ depth for community $K_c$, and $\boldsymbol{u}_{t}$ are IID white noise with $\mathrm{cov} (\boldsymbol{u}_t) = \sigma^2 \bmat{I_d}$ and $\sigma^2 > 0$. We express the sum of stationary community processes as a \mbox{local-$\alpha$} GNAR process as follows.
\par
Let $\alpha_{i, k} = \alpha_{k_c, c} \mathbb{I} (i \in K_c)$, where we set $\alpha_{k_c, c} = 0$ if $k_c > p_c$, and $\beta_{i, k, r} = \beta_{k_c, r, c} \mathbb{I} (i \in K_c)$, where we set $\beta_{k_c, r, c} = 0$ if $k_c > p_c$ or $ s_{k_c} > r^*_c$. Hence, the community with the largest lag will have the biggest number of non-zero autoregressive coefficients, and might have the most non-zero $\rstage$ neighbourhood coefficients. By construction we have that for all $X_{i, t}$ such that $i \in K_c$ the $\rstage$ neighbourhood coefficients are equal for all lags and at all $\rstage$s (i.e., if $i \in K_c$, then $\beta_{i, k, r} = \beta_{k_c, r, c}$ otherwise $\beta_{i, k, r} = 0$). Set $p = \max ( [p_c])$ and $s_k = \max ([s_{k_c}]) $ for each $k = 1, \dots, p$. Then, the node-wise representation of the $\comGNAR$ GNAR model is given by
\begin{equation} \label{eq: node wise com ganr}
    X_{i, t} = \sum_{k = 1}^{p} \left ( \alpha_{i, k} X_{i, t - k} + \sum_{r = 1}^{s_k} \beta_{i, k, r} Z_{i, t - k}^{r, c} \right ) 
    + u_{i, t},
\end{equation}
where $u_{i, t}$ are IID white noise and $Z_{i, t - k}^{r, c}$ are the $c$-community $\rstage$ neighbourhood regressions. For all lags $k \in \{1, \dots, p\}$ we have that for each $r \in \{1, \dots, s_k\}$ the node-wise $\rstage$ neighbourhood regression coefficients for $X_{i, t}$ satisfy
$$
    \sum_{r = 1}^{s_k} |\beta_{i, k, r} | = \sum_{r = 1}^{s_k} |\beta_{k_c, r, c}| \mathbb{I}(i \in K_c),
$$
and that the autoregressive coefficients satisfy
$$
    \sum_{k = 1}^p |\alpha_{i, k}| = \sum_{k = 1}^p |\alpha_{k_c, p}| \mathbb{I} (i \in K_c), 
$$
thus, the coefficients for each nodal time series $X_{i, t}$ satisfy
$$
    \sum_{k = 1}^p \left ( |\alpha_{i, k}| +  \sum_{r = 1}^{s_k} |\beta_{i, k, r} | \right ) = \sum_{k = 1}^p \left \{ |\alpha_{k_c, p}| \mathbb{I} (i \in K_c) + \sum_{r = 1}^{s_k} |\beta_{k_c, r, c}| \mathbb{I}(i \in K_c) \right \}.
$$
Assume that $\sum_{k = 1}^{p_c} \left ( |\alpha_{k_c, c} | + \sum_{r = 1}^{s_{k_c}} |\beta_{k_c, r, c} | \right ) < 1$ for all covariates $c \in [C]$. Combining this with the above we have that
$$
     \sum_{k = 1}^p \left ( |\alpha_{i, k}| +  \sum_{r = 1}^{s_k} |\beta_{i, k, r} | \right ) < 1
$$
for all $i \in \{1, \dots, d\}$. Therefore, by Theorem \ref{th: stationary cond} the process given by \eqref{eq: node wise com ganr} is stationary. Recalling that \eqref{eq: node wise com ganr} is identical to the process given by \eqref{eq: community GNAR structural representation} finishes the proof.
\end{proof}

\subsection*{Supplementary plots}
\subsubsection*{Supplementary plots for the presidential election network time series}
\begin{figure}[H]
    \centering
    \includegraphics[scale=0.4]{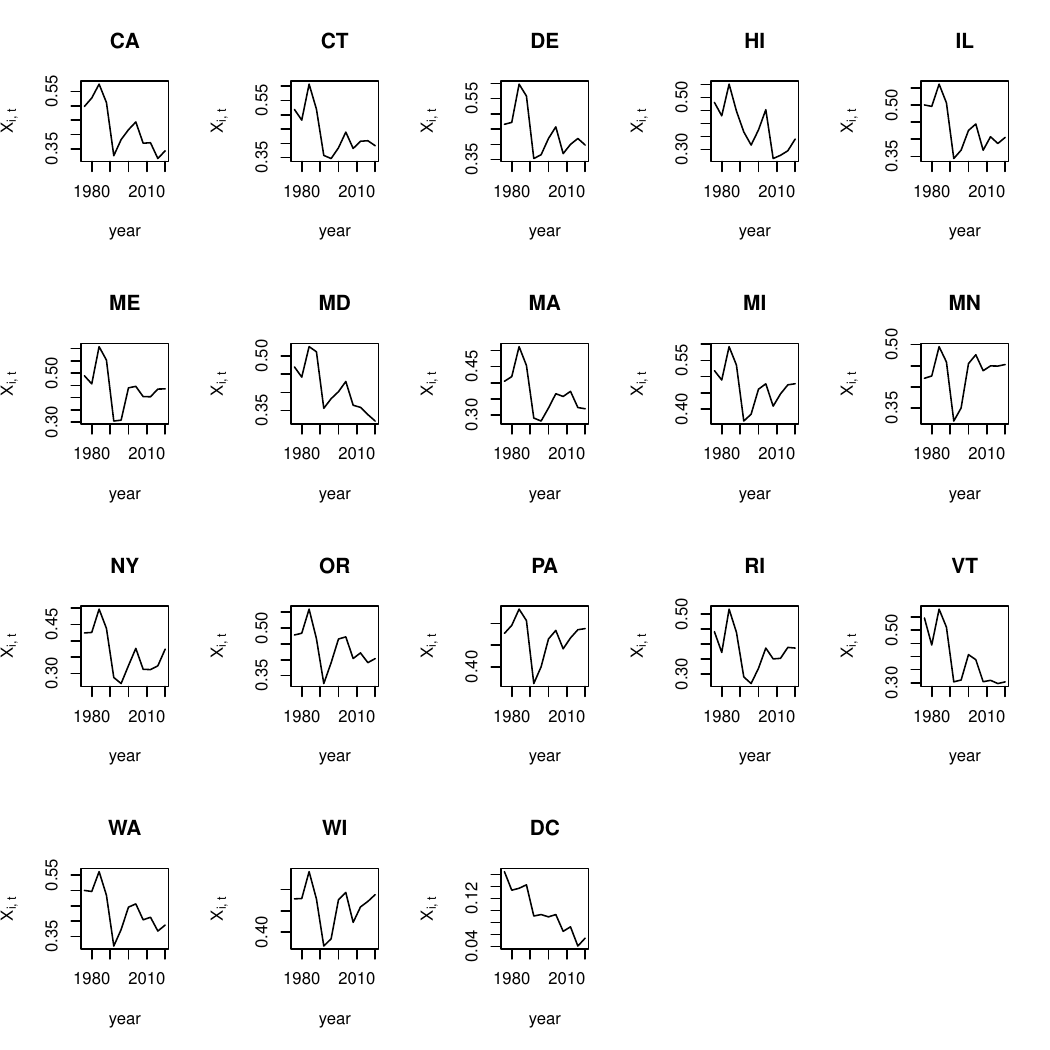}
    \caption{{\em Blue} states from 1976 to 2020.}
    \label{fig:Blue states}
\end{figure}

\begin{figure}[H]
    \centering
    \includegraphics[scale=0.4]{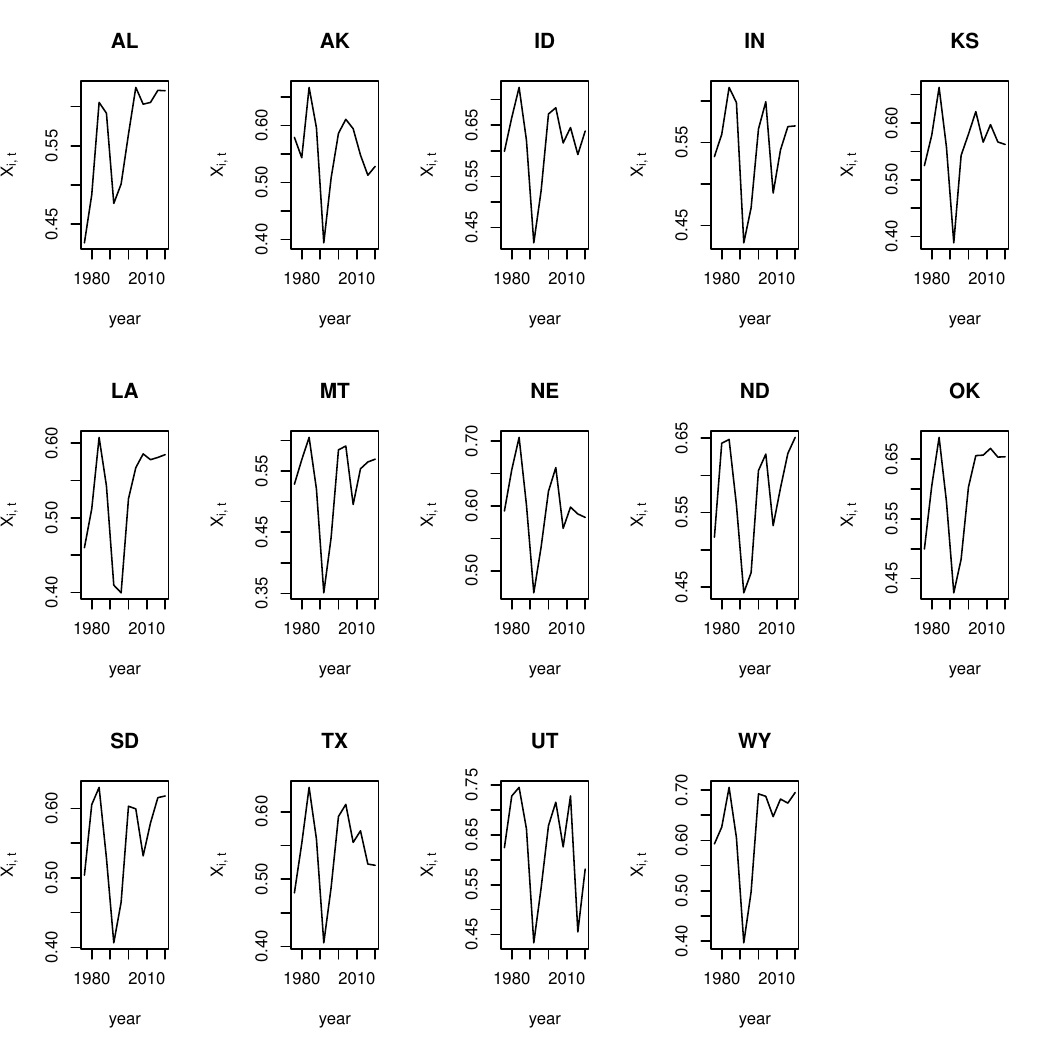}
    \caption{{\em Red} states from 1976 to 2020.}
    \label{fig: Red states}
\end{figure}

\begin{figure}[H]
    \centering
    \includegraphics[scale=0.4]{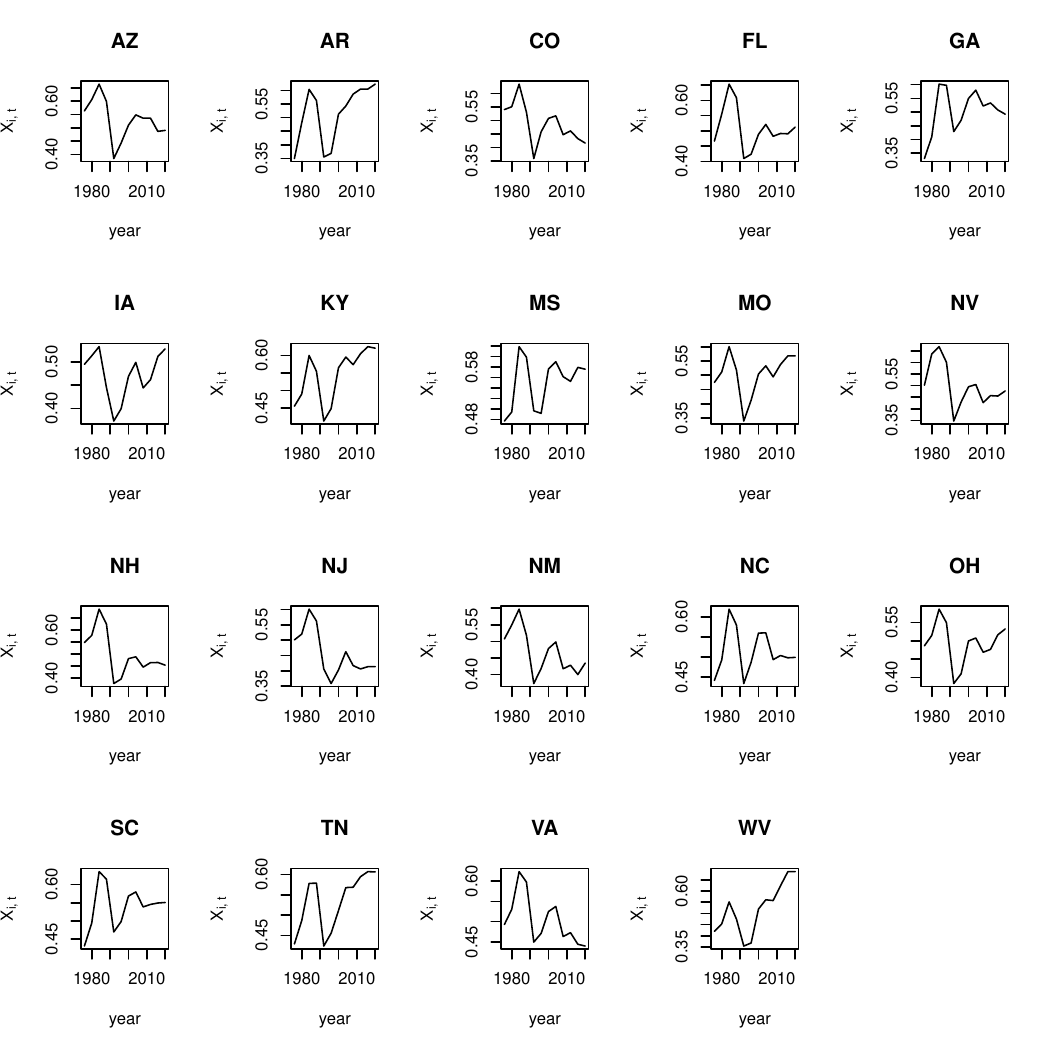}
    \caption{{\em Swing} States from 1976 to 2020.}
    \label{fig: swing states}
\end{figure}

\begin{figure}[H]
    \centering
    \includegraphics[scale=0.35]{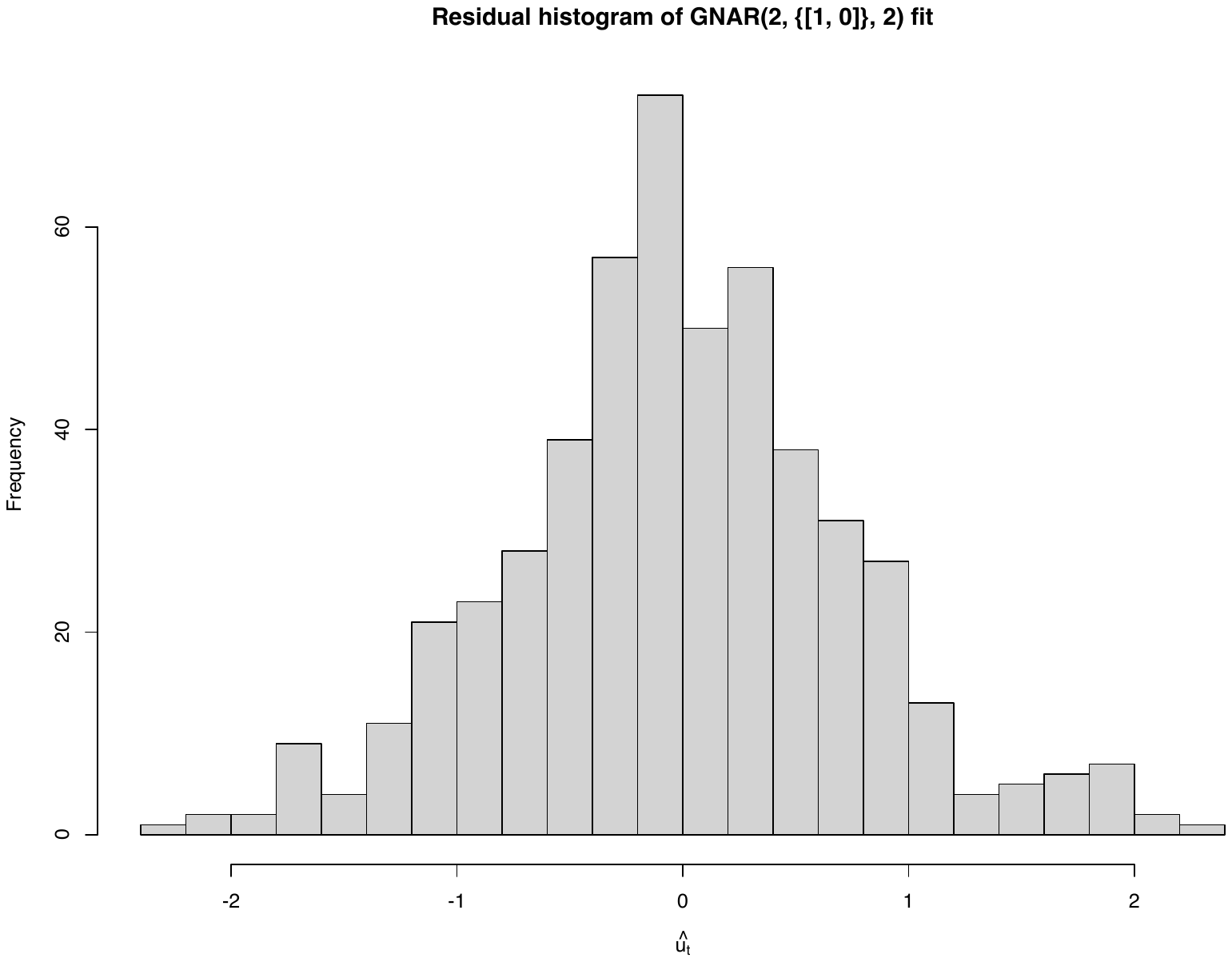}
    \caption{Residual histogram for a $\GNAR(2, \{[1, 0]\}, 3)$ fit to the presidential election network time series.}
    \label{fig: residual non-differenced}
\end{figure}

\begin{figure}[H]
    \centering
    \includegraphics[scale=0.35]{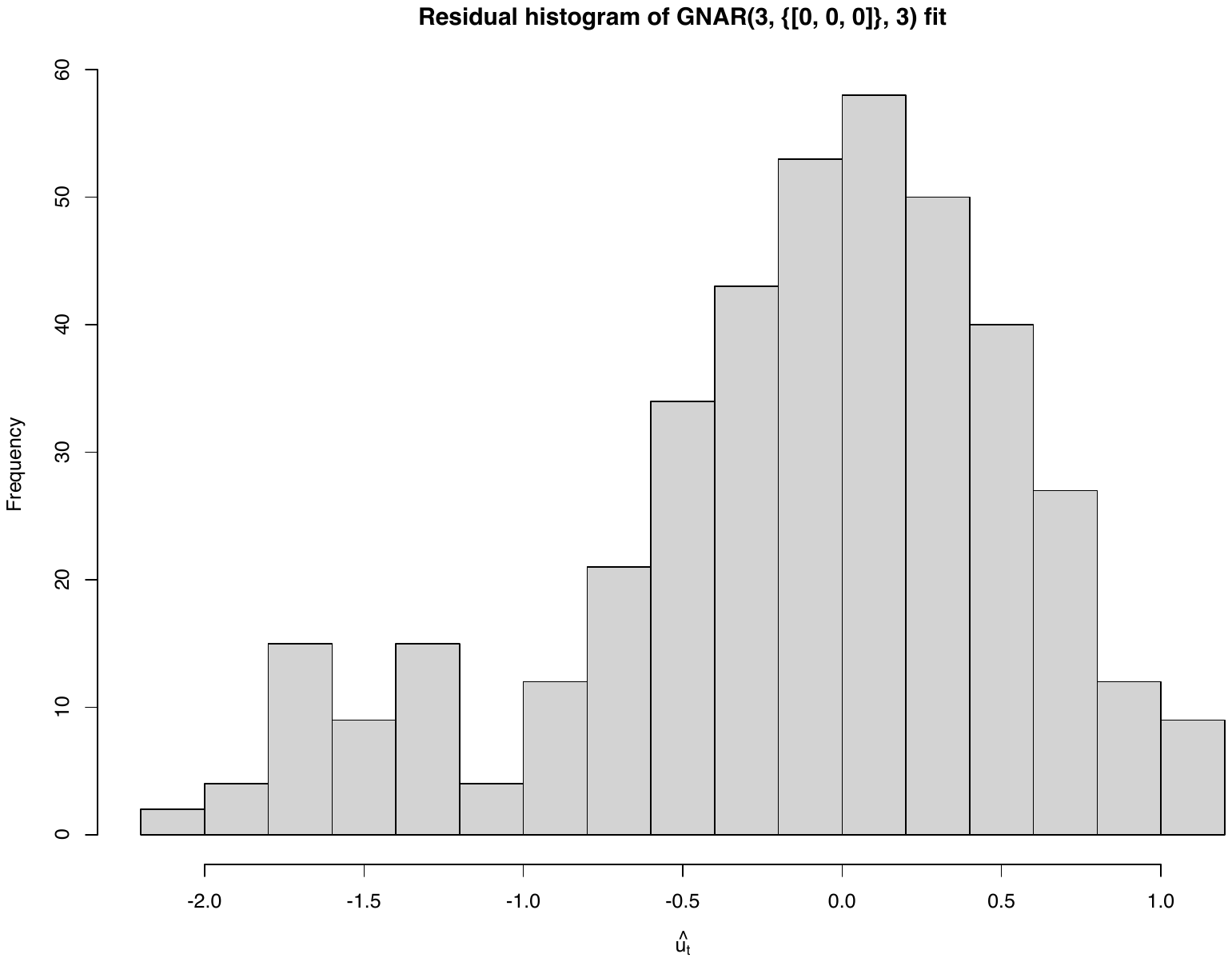}
    \caption{Residual histogram for a $\GNAR(3, \{[0, 0, 0]\}, 3)$ fit to the one-lag differenced presidential election network time series.}
    \label{fig: residual differenced}
\end{figure}
Interestingly, the residuals for the non-differenced data are symmetric around zero and appear to decay in a bell-shape. Also, the residuals for the differenced data are not symmetric, however, it appears that some outliers are making the residual distribution asymmetric. The histogram in Figure \ref{fig: residual differenced} resembles a bell-curve if we omit values that might be outliers (i.e., $|\hat{u}_t| > 1.35$). A more careful analysis might reveal interesting properties.

\begin{figure}[H]
        \centering
        \begin{subfigure}{0.48\textwidth}
            \includegraphics[width=\textwidth]{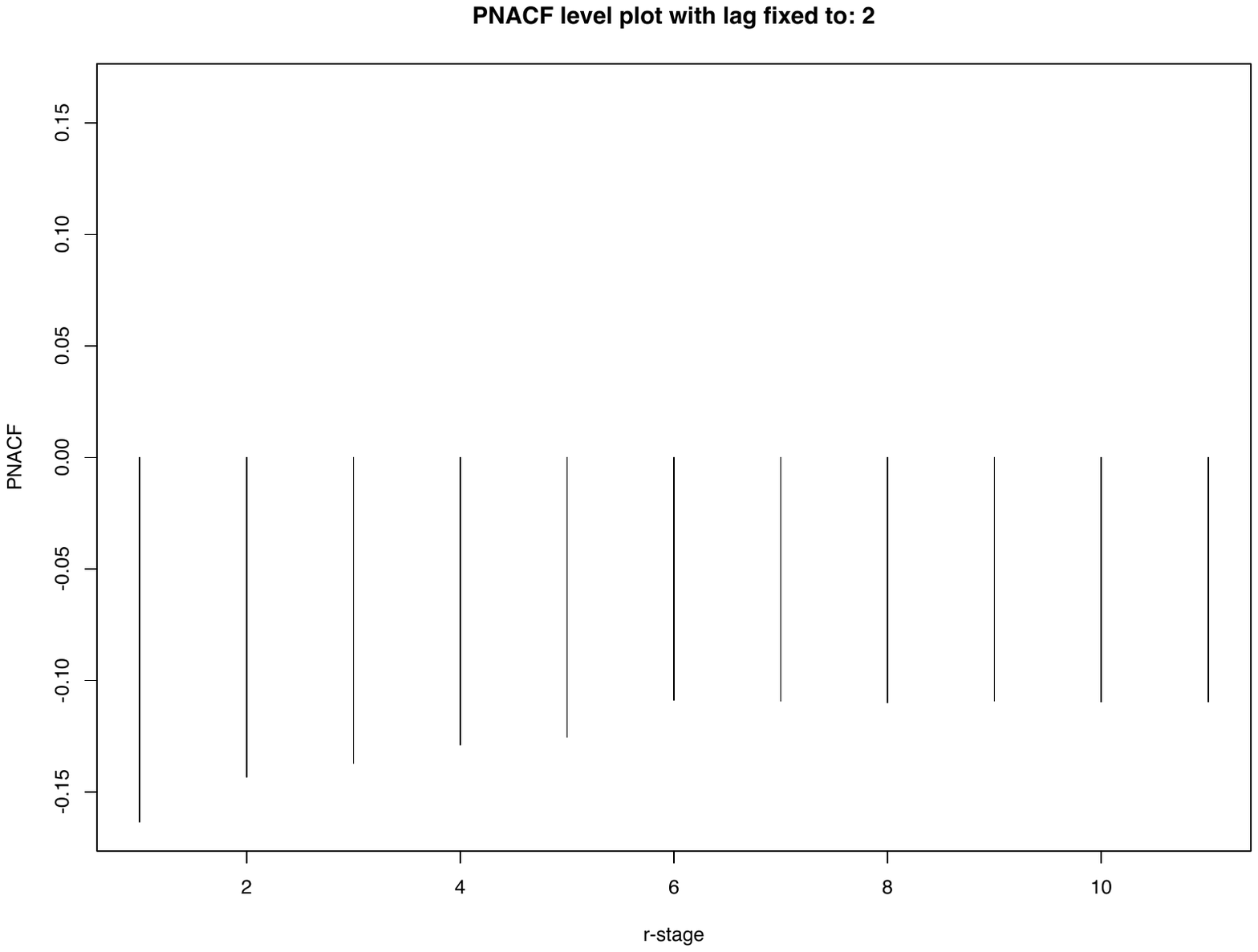}
            \caption{}
            \label{fig: pnacf level red lag 2}
        \end{subfigure}
        \hfill
        \begin{subfigure}{0.48\textwidth}
            \includegraphics[width=\textwidth]{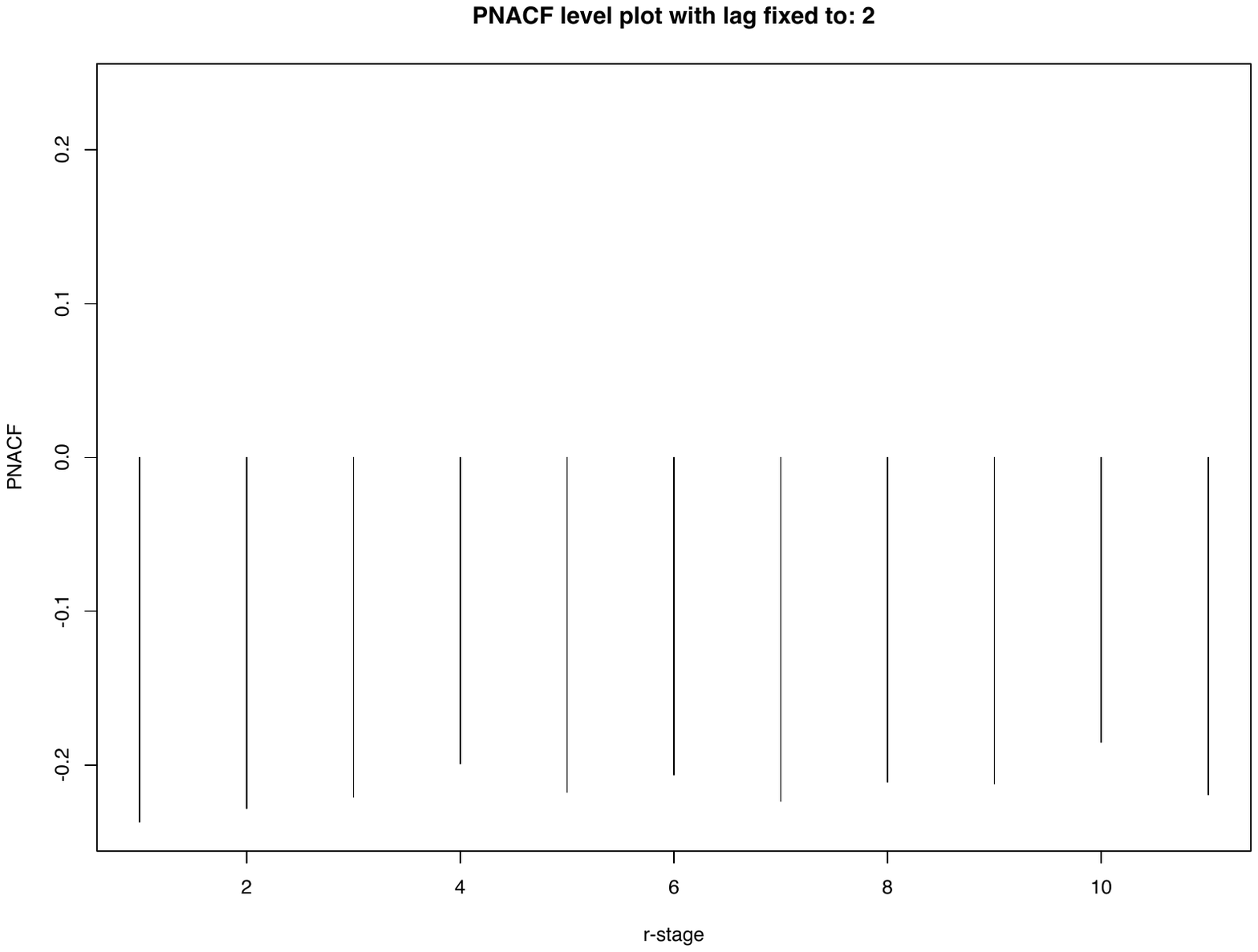}
            \caption{}
            \label{fig: pnacf level blue lag 2}
        \end{subfigure}
        \hfill
        \begin{subfigure}{0.48\textwidth}
            \includegraphics[width=\textwidth]{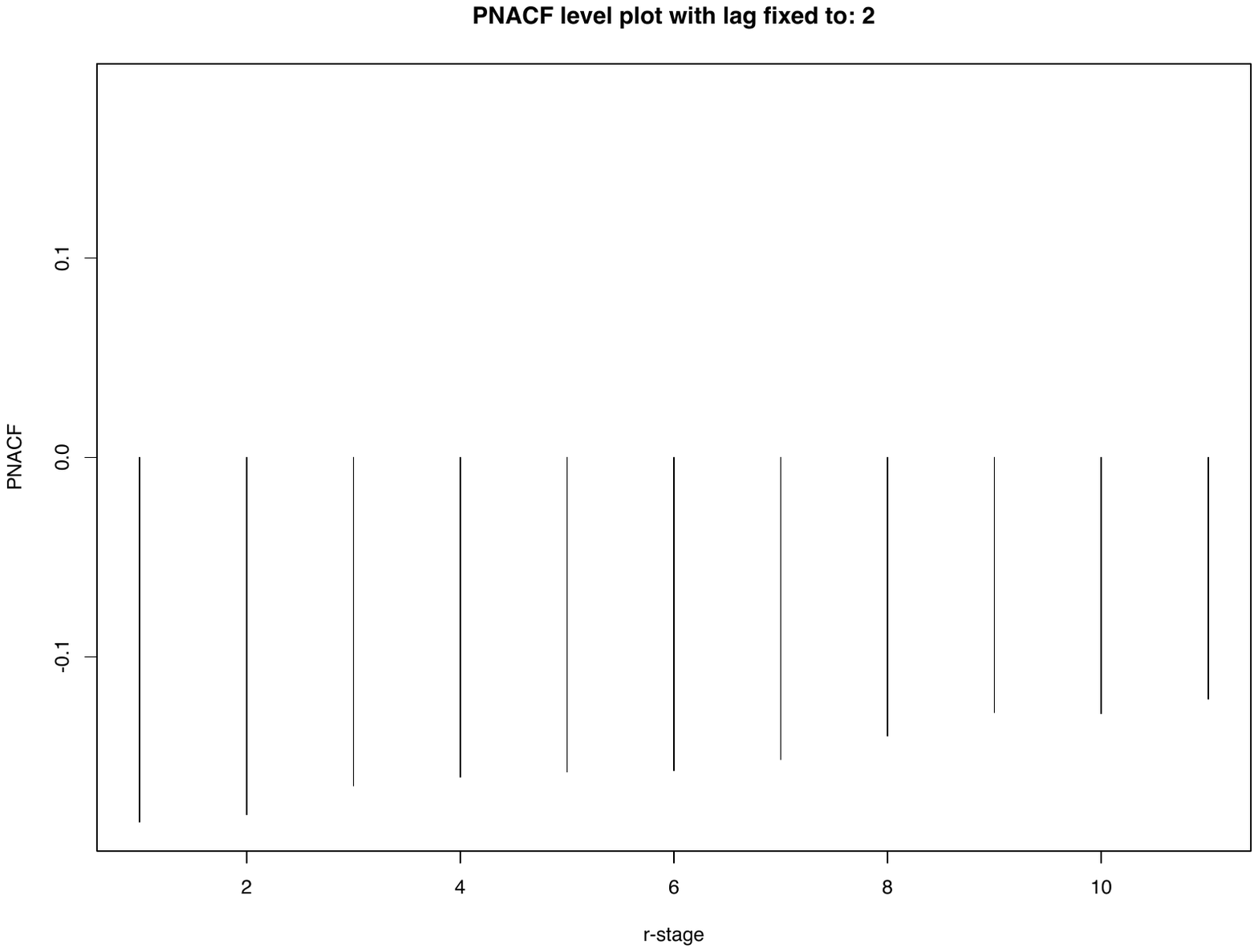}
            \caption{}
            \label{fig: pnacf level swing lag 2}
        \end{subfigure}
        \hfill
        \caption{
            PNACF level plots with fixed lag, each plot shows the PNACF for $\rstage$s one thru eleven with lag fixed to two. Figure \ref{fig: pnacf level red lag 2} is for lag fixed to one and {\em Red} states. Figure \ref{fig: pnacf level blue lag 2} is for lag fixed to one and {\em Blue} states. Figure \ref{fig: pnacf level swing lag 2} is for lag fixed to two and {\em Swing} states. The data are the 11 one-lag differenced presidential elections in the USA from 1980 to 2020; see Figure \ref{fig: usa net}.
        }
        \label{fig: pnacf USA level plots}
\end{figure}

\subsubsection*{Supplementary plots for the simulation in Section \ref{sec: model selection}}
\begin{figure}[H]
        \centering
        \begin{subfigure}{0.48\textwidth}
            \includegraphics[width=\textwidth]{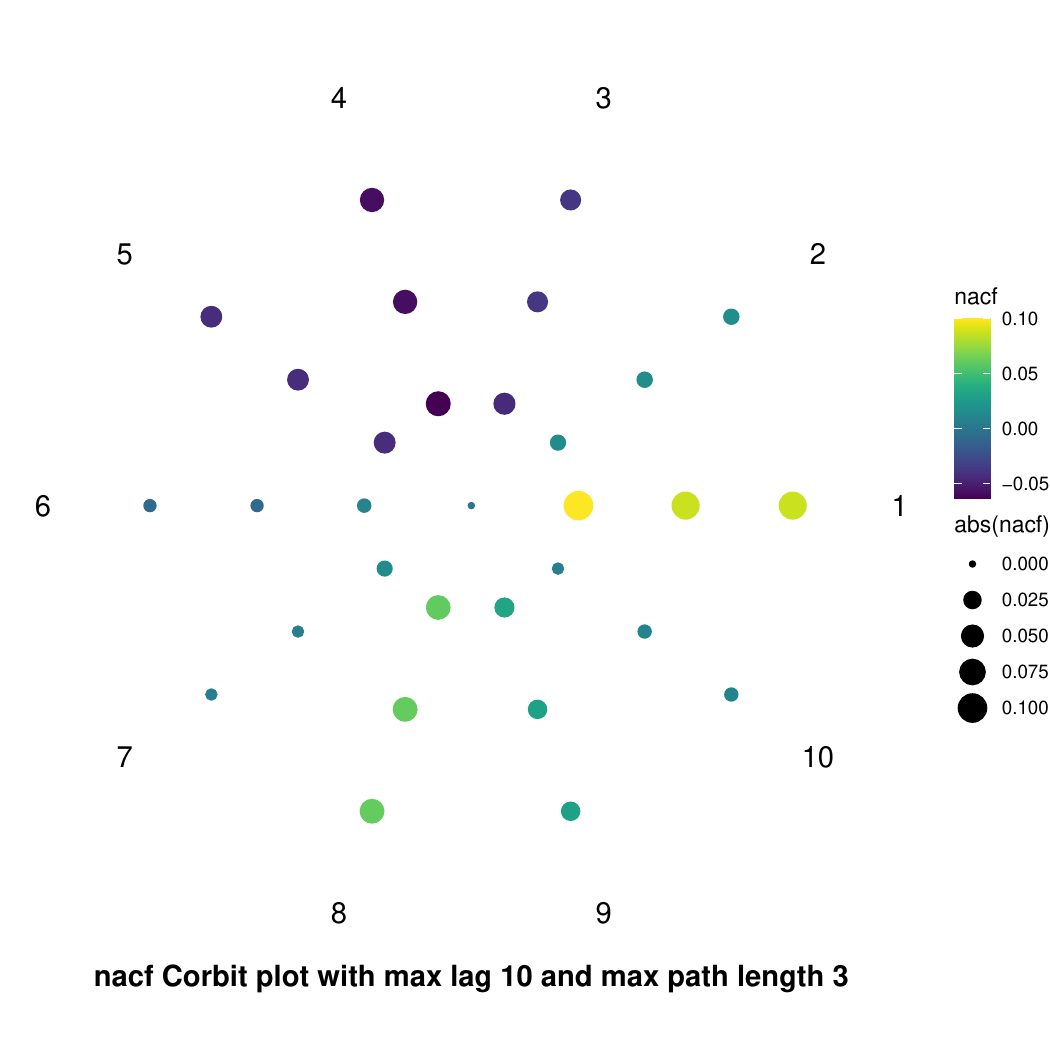}
            \caption{}
            \label{fig: corbit nacf 1}
        \end{subfigure}
        \hfill
        \begin{subfigure}{0.48\textwidth}
            \includegraphics[width=\textwidth]{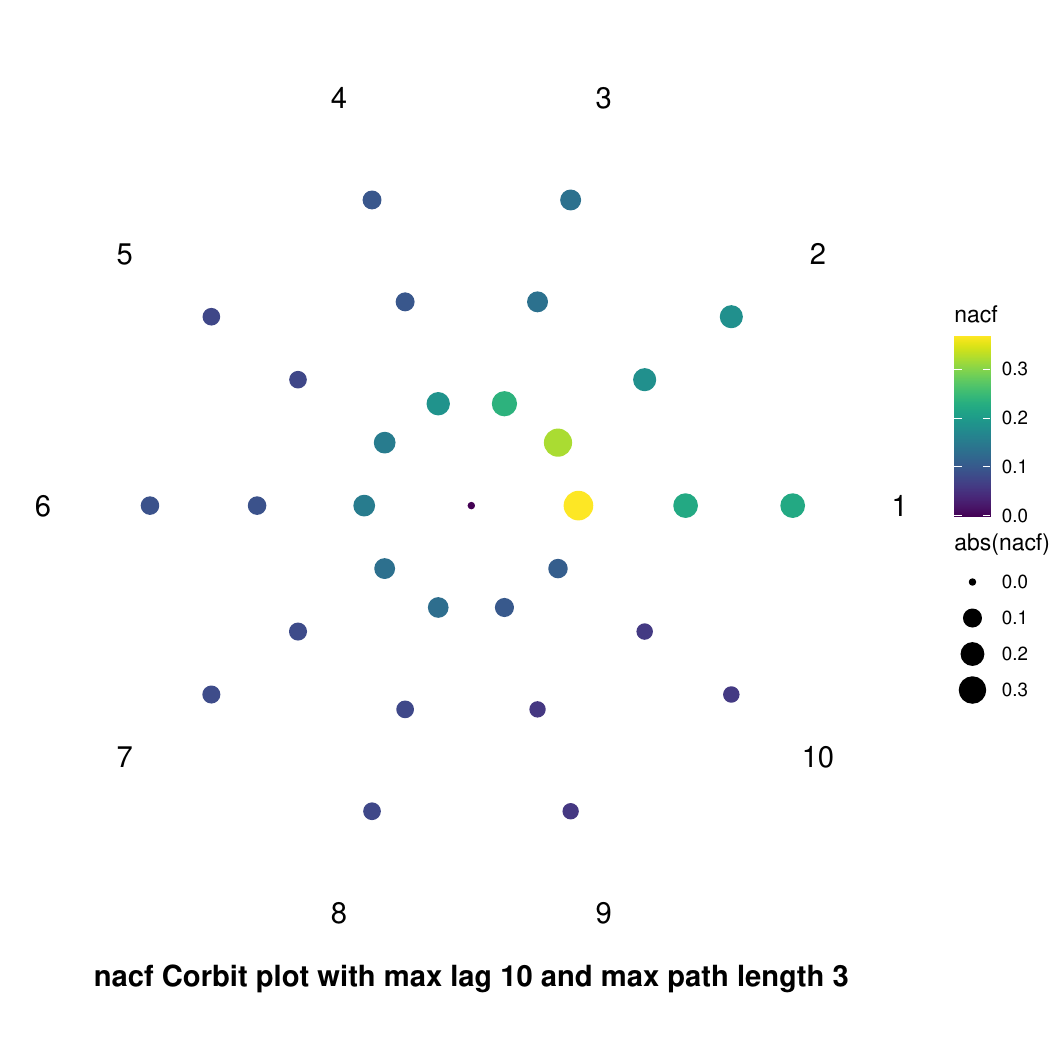}
            \caption{}
            \label{fig: corbit nacf 2}
        \end{subfigure}
        \hfill
        \begin{subfigure}{0.48\textwidth}
            \includegraphics[width=\textwidth]{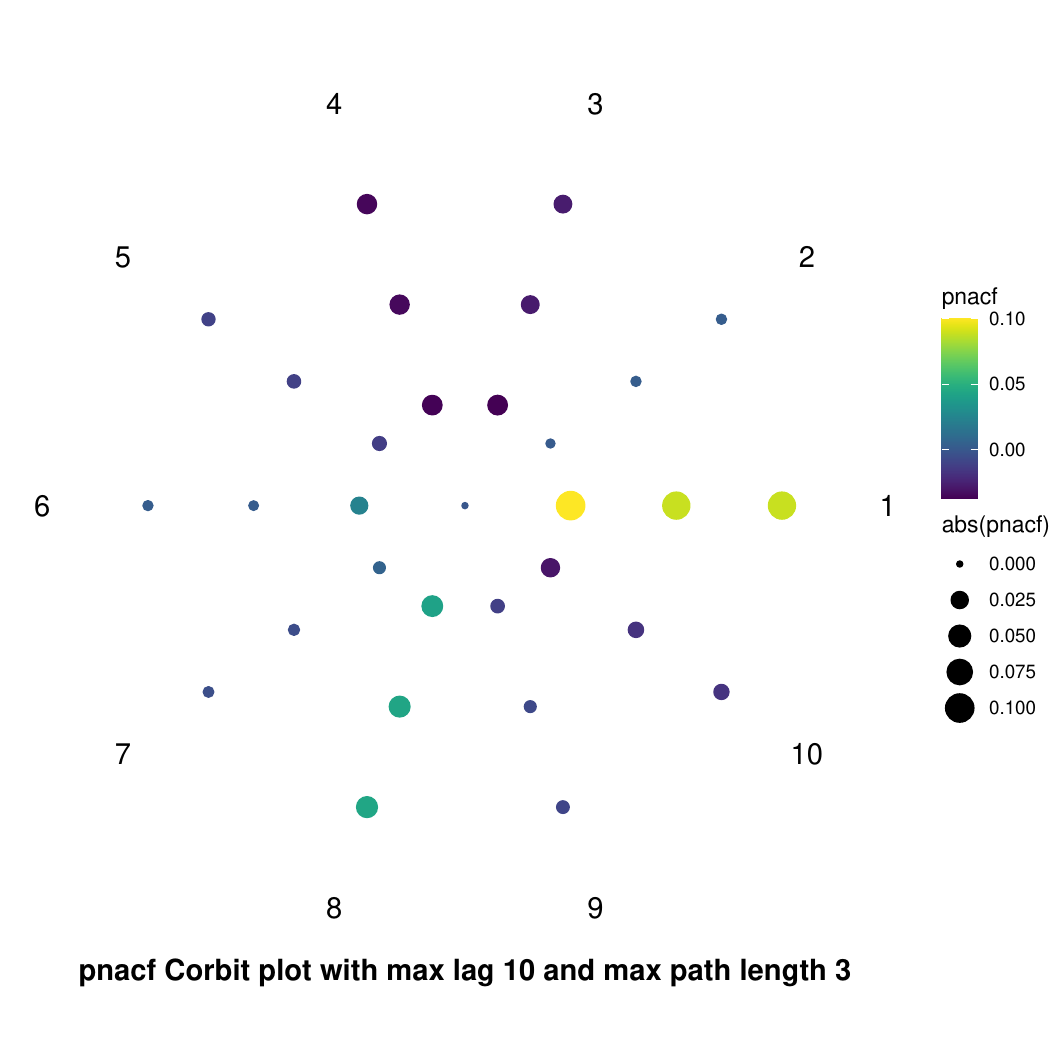}
            \caption{}
            \label{fig: corbit pnacf 1}
        \end{subfigure}
        \hfill
        \begin{subfigure}{0.48\textwidth}
            \includegraphics[width=\textwidth]{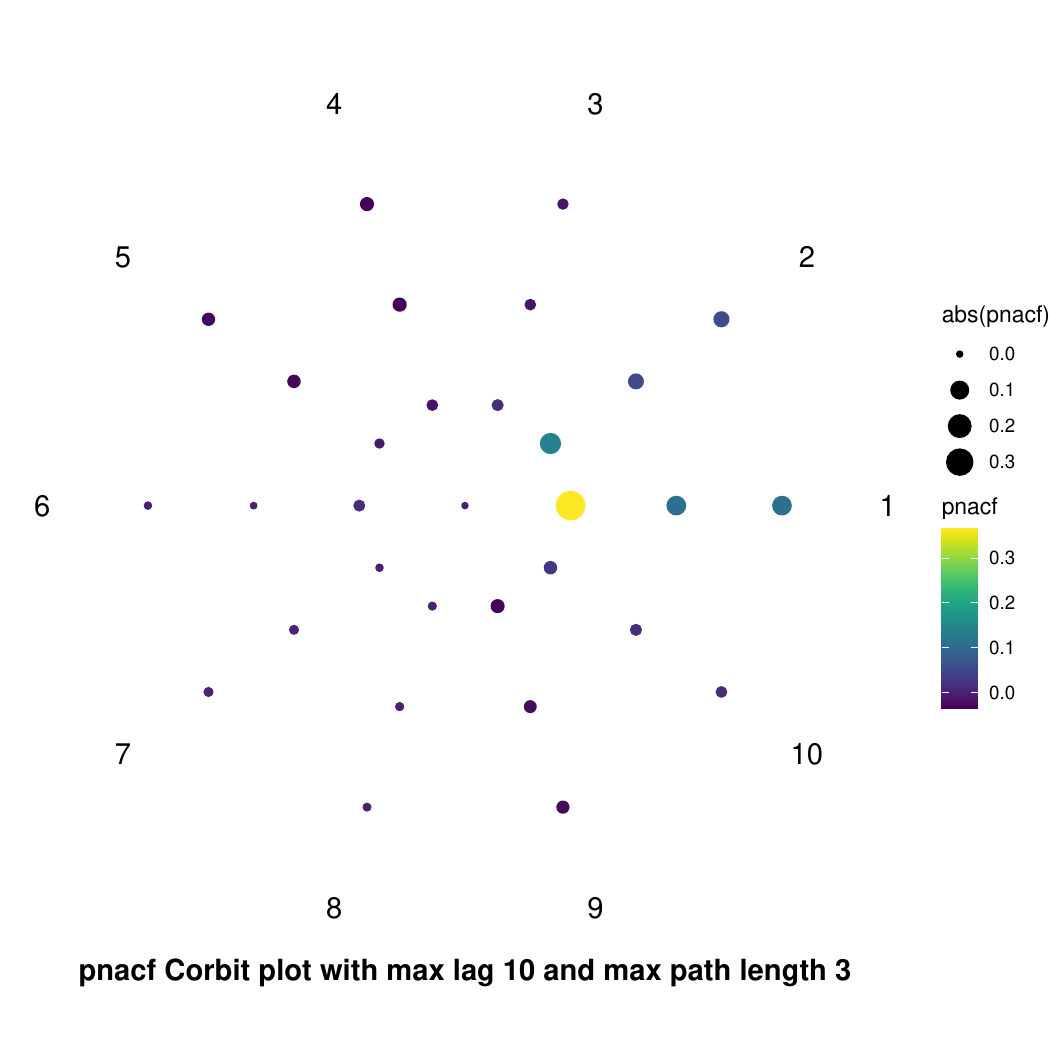}
            \caption{}
            \label{fig: corbit pnacf 2}
        \end{subfigure}
        \caption{Corbit plots: Figure \ref{fig: corbit nacf 1} shows the NACF for community one, Figure \ref{fig: corbit nacf 2} shows the NACF for community two, Figure \ref{fig: corbit pnacf 1} shows the PNACF for community one, and Figure \ref{fig: corbit pnacf 2} shows the PNACF for community two. The data are 100 realisations coming from a stationary $\comGNAR$ $\GNAR \left ( [1, 2], \{ [1], [1, 1] \}, 2 \right )$, where the underlying network is \textbf{fiveNet}, $K_1 = \{2, 3, 4\}$ and $K_2 = \{1, 5 \}$; see Figure \ref{fig: 2-communal fiveNet}.}
        \label{fig: corbits}
\end{figure}

\begin{figure}[H]
        \centering
        \begin{subfigure}{0.48\textwidth}
            \includegraphics[width=\textwidth]{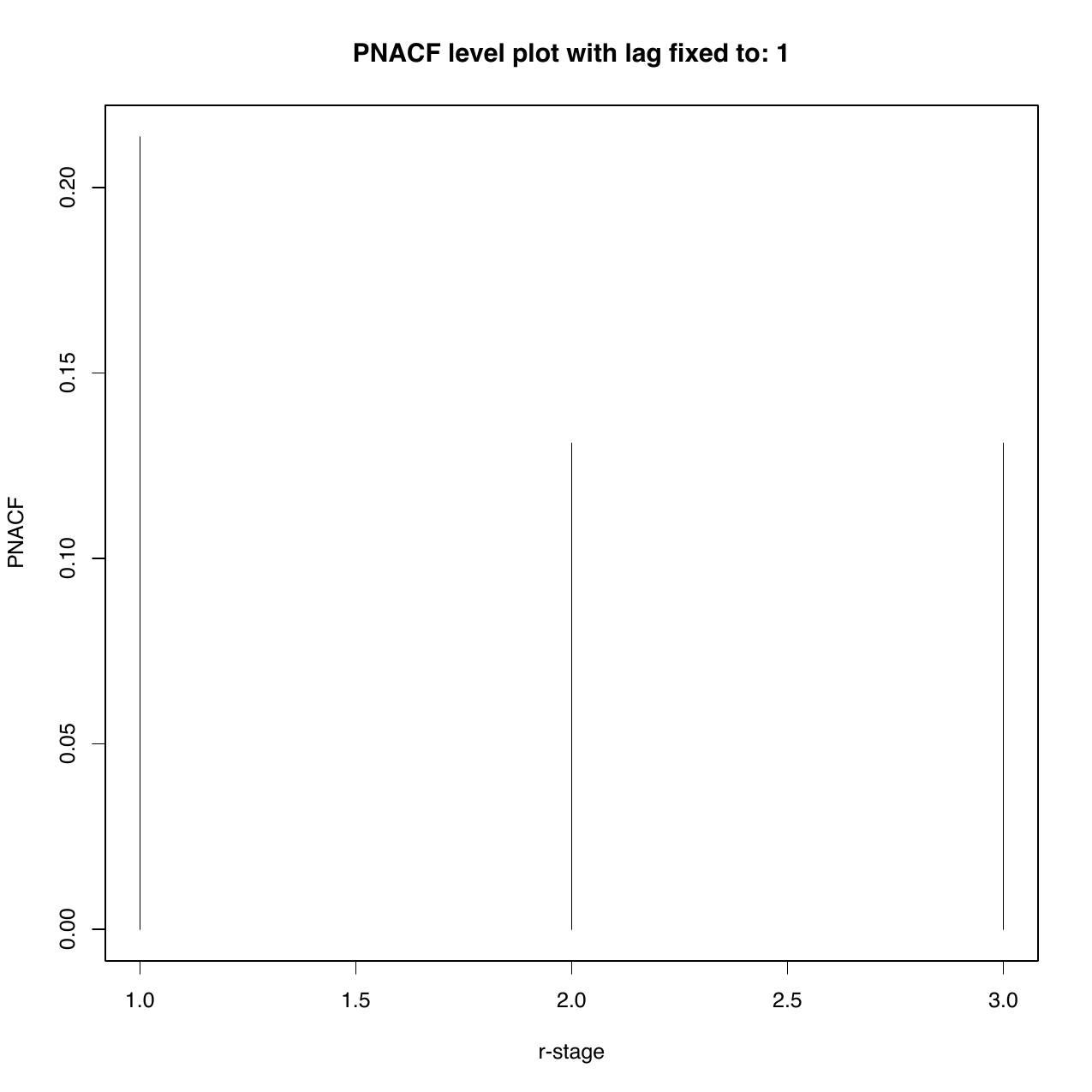}
            \caption{}
            \label{fig: pnacf level one lag 1}
        \end{subfigure}
        \hfill
        \begin{subfigure}{0.48\textwidth}
            \includegraphics[width=\textwidth]{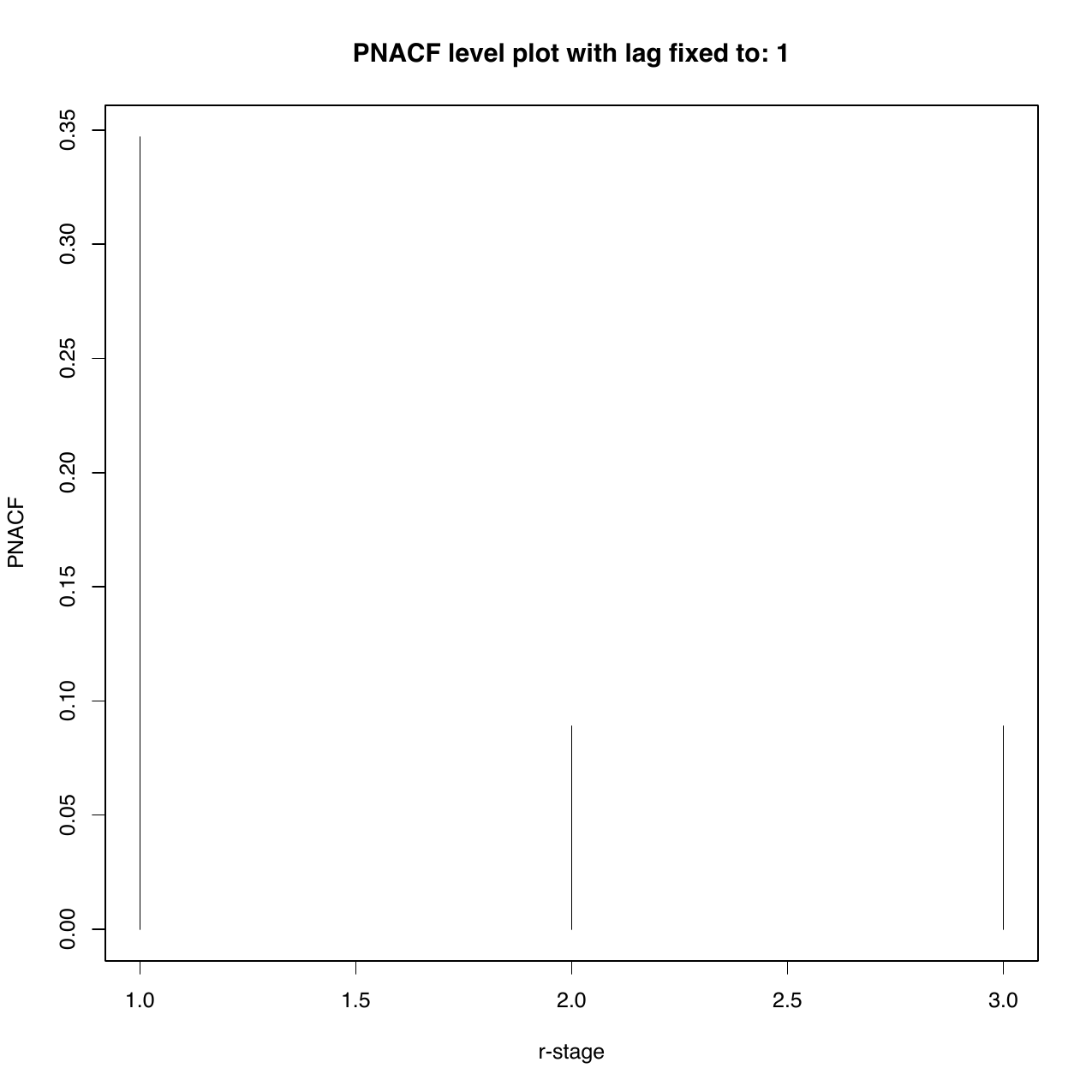}
            \caption{}
            \label{fig: pnacf level two lag 1}
        \end{subfigure}
        \hfill
        \begin{subfigure}{0.48\textwidth}
            \includegraphics[width=\textwidth]{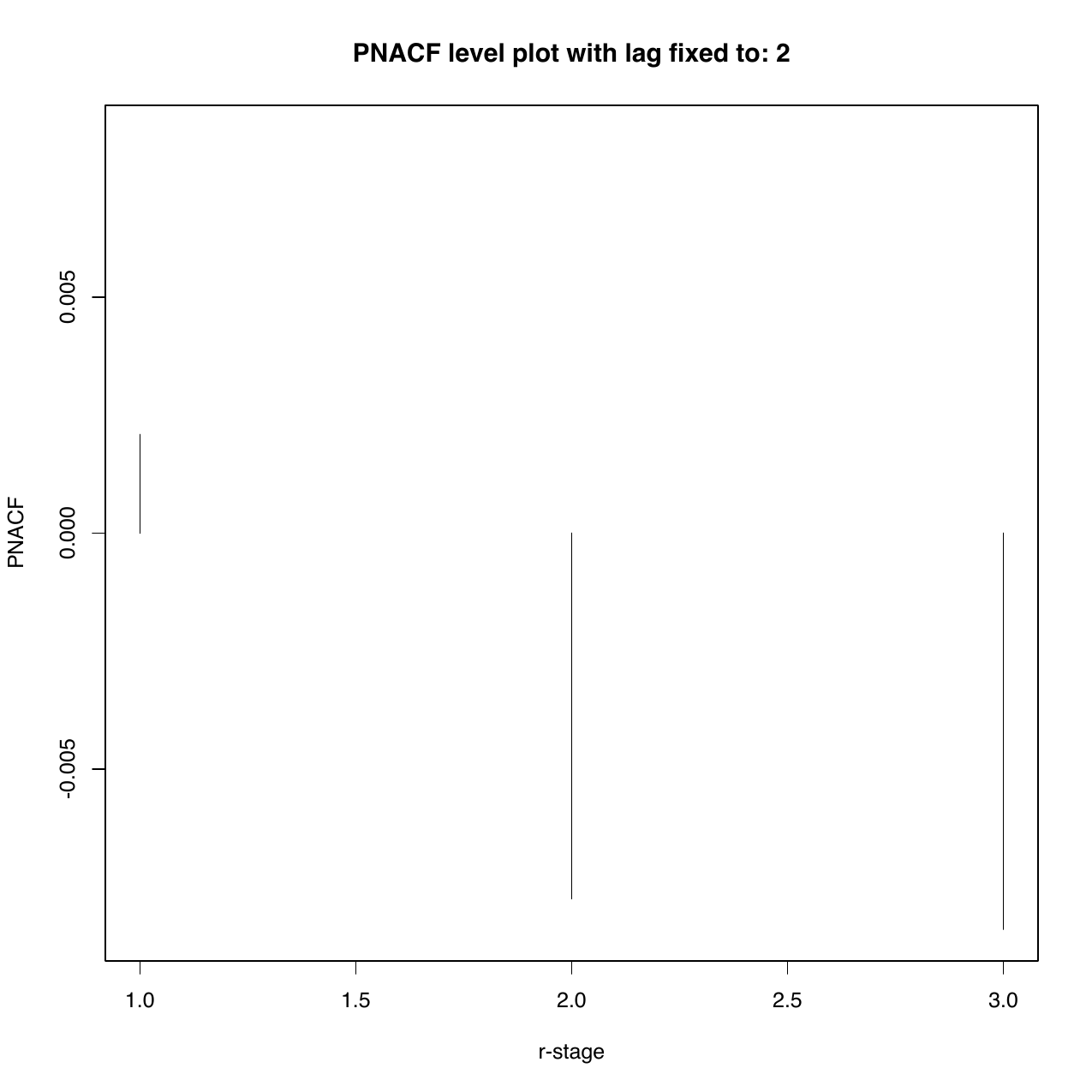}
            \caption{}
            \label{fig: pnacf level one lag 2}
        \end{subfigure}
        \hfill
        \begin{subfigure}{0.48\textwidth}
            \includegraphics[width=\textwidth]{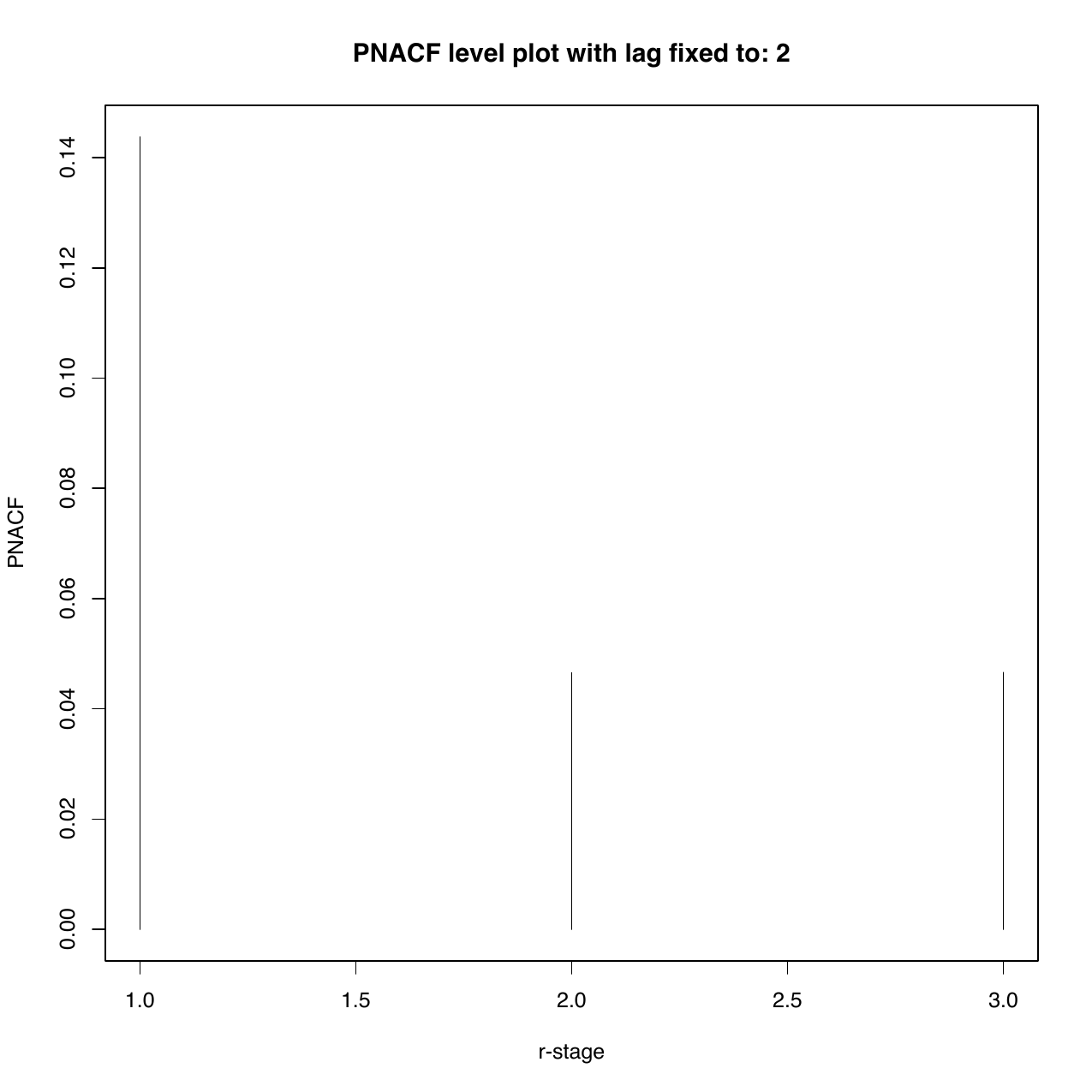}
            \caption{}
            \label{fig: pnacf level two lag 2}
        \end{subfigure}
        \caption{
            PNACF level plots with fixed lag, each plot shows the PNACF for $\rstage$s one thru three with lag fixed to one or two. Figure \ref{fig: pnacf level one lag 1} is for lag fixed to one and community one. Figure \ref{fig: pnacf level two lag 1} is for lag fixed to one and community two. Figure \ref{fig: pnacf level one lag 2} is for lag fixed to two and community one. And, Figure \ref{fig: pnacf level two lag 2} is for lag fixed to two and community two. The data are 100 realisations coming from a stationary $\comGNAR$ $\GNAR \left ( [1, 2], \{ [1], [1, 1] \}, 2 \right )$, where the underlying network is \textbf{fiveNet}, $K_1 = \{2, 3, 4\}$ and $K_2 = \{1, 5 \}$; see Figure \ref{fig: 2-communal fiveNet}.
        }
        \label{fig: pnacf level plots}
\end{figure}